\documentclass[10pt,journal]{IEEEtran}

\usepackage{amssymb}
\usepackage{amsmath}
\usepackage{cite}
\usepackage{url}
\usepackage{xcolor}
\usepackage{cite,graphicx,amsmath,amssymb}
\usepackage{fancyhdr}
\usepackage{mdwmath}
\usepackage{mdwtab}
\usepackage{caption}
\usepackage{amsthm}
\usepackage{setspace}
\usepackage{hyperref}
\usepackage{algorithm}
\usepackage{algorithmic}
\usepackage{multirow}
\usepackage{makecell}
\usepackage{mathtools}
\usepackage{subcaption}
\usepackage{bm}
\usepackage{tikz}
\usepackage{xurl}
\usepackage{stfloats}
\usepackage{mathrsfs}

\usepackage{booktabs}
\usepackage{multirow}
\usepackage{siunitx}
\usepackage{balance}

\hypersetup{colorlinks=true,
linkcolor=blue,
citecolor=blue,      
urlcolor=black,
}

\graphicspath{{./src/img/}}

\DeclareMathOperator*{\argmax}{arg\,max}

\newtheorem{theorem}{Theorem}

\newtheorem{lemma}{Lemma}

\newtheorem{assumption}{Assumption}

\newtheorem{corollary}{Corollary}

\allowdisplaybreaks
\NewDocumentCommand{\multiubrace}{mmm}
 {
  \egreg_multiubrace:nnn {#1} {#2} {#3}
 }

\expandafter\def\expandafter\normalsize\expandafter{%
	\normalsize%
	\setlength\abovedisplayskip{4pt}%
	\setlength\belowdisplayskip{4pt}%
	\setlength\abovedisplayshortskip{2pt}%
	\setlength\belowdisplayshortskip{2pt}%
}
\begin{document}
\title{Bridging Standardized Codebook and Site-Specific Beamforming: A Unified Limited-Feedback Framework} 
\author{{Cheng-Jie Zhao, Zhaolin Wang,~\IEEEmembership{Member, IEEE}, Zongyao Zhao,~\IEEEmembership{Member, IEEE}, and Yuanwei Liu,~\IEEEmembership{Fellow, IEEE}}
	\vspace{-0.9cm}
\thanks{The authors are with Department of Electrical and Computer Engineering, The University of Hong Kong, Hong Kong (e-mail: chengjie\_zhao@connect.hku.hk; zhaolin.wang@hku.hk; zongyao@hku.hk; yuanwei@hku.hk).}
}
\maketitle

\begin{abstract}
A site-specific Type-II codebook design is proposed for downlink massive multiple-input multiple-output (MIMO) limited-feedback beamforming. The key idea is to embed a learned site-specific propagation prior into the Type-II channel state information (CSI) feedback pipeline. Specifically, the base station (BS) uses a low-overhead reference signal received power (RSRP) fingerprint collected during synchronization signal block (SSB) probing to infer a user equipment (UE)-dependent dominant beam subspace before explicit CSI acquisition. The UE then estimates and feeds back only the low-dimensional effective channel coefficients within this inferred subspace, thereby avoiding full-dimensional online subspace discovery while retaining a rich multi-beam representation capability. To analyze the proposed design and compare it with standardized feedback mechanisms, a unified subspace-projection framework is developed by jointly characterizing CSI acquisition, UE-side compression, BS-side reconstruction, and effective spectral efficiency. Under this framework, Type-I, Type-II, port-selection feedback, and the proposed scheme are interpreted as different ways of inducing a feedback representation subspace. The probing codebook and the BS-side subspace inference network are then formulated as a coupled task-oriented design problem and are optimized end-to-end by maximizing the normalized CSI-capture efficiency. Extensive simulation results demonstrate that the proposed feedback scheme achieves Type-II-comparable CSI-capture capability with substantially lower online overhead and UE-side complexity, thereby improving the effective spectral efficiency.
\end{abstract}

\vspace{-0.5cm}
\section{Introduction}
\IEEEPARstart{B}EAMFORMING is a key enabling technology for modern cellular systems, especially in massive multiple-input multiple-output (MIMO) and high-frequency deployments where directional transmission is indispensable for compensating path loss and improving spectral efficiency \cite{Heath2016mmWave}. In practical downlink (DL) systems, however, high-quality beamforming critically depends on how accurately the base station (BS) can acquire the user equipment (UE)-specific channel state information (CSI). This challenge is particularly pronounced in frequency-division-duplex (FDD) systems and in scenarios where reciprocity-based acquisition is imperfect or unavailable \cite{Hassibi2003Training}. In such cases, the DL CSI must be learned through BS-side reference-signal (RS) transmission, UE-side channel estimation and compression, and subsequent CSI feedback. The resulting overhead can become substantial when the antenna dimension is large and the channel varies rapidly, which makes limited-feedback beamforming a fundamental yet challenging design problem \cite{Love2008LimitedFeedback}. \\
\indent Current 3rd Generation Partnership Project (3GPP) New Radio (NR) systems address this problem through standardized codebook-based CSI feedback mechanisms specified in TS 38.214 \cite{3GPP38214}, among which Type-I, Type-II, and enhanced Type-II port-selection codebook (PSC) schemes constitute the main design options. These schemes offer different tradeoffs between signaling overhead, UE-side complexity, and beamforming accuracy. Type-I relies on single-beam quantization and therefore has the lowest feedback and computational cost, but its representation capability is limited in multipath environments. Type-II provides a more flexible multi-beam subspace representation and often achieves much better beamforming quality, but this gain comes at the price of full-dimensional CSI acquisition, heavier beam search, and higher UE-side processing and feedback overhead. PSC provides a structured alternative, but its representation capability is still constrained by the selected port domain. As the core codebook-based feedback schemes in NR, they have been extensively discussed in the broader literature on limited feedback, beam management, and 3GPP MIMO codebooks \cite{Love2008LimitedFeedback,Giordani2019BeamManagementNR,Fu2023TutorialCodebooks}. \\
\indent Existing studies have mainly followed two separate lines. One line remains within the standardized or near-standardized codebook-based feedback framework and seeks to improve beam management or CSI feedback while preserving compatibility with practical NR signaling structures and feedback formats. In particular, the authors of \cite{Dreifuerst2024MLCodebook} proposed a machine-learning-based codebook design for NR initial access and Type-II feedback . In a subsequent work, they further developed a neural codebook design framework for broader MIMO network beam-management tasks \cite{Dreifuerst2025NeuralCodebook}. These works show that data-driven codebook design can improve practical beam management under standardized signaling constraints. However, they still treat Type-I, Type-II, and PSC largely within their individual formulations, and therefore do not provide a common analytical lens for systematically exposing the intrinsic overhead-complexity-performance tradeoffs among the conventional feedback mechanisms. \\
\indent A separate line of research moves toward site-specific or environment-aware beamforming, where learned propagation structure or site-specific measurements are exploited for probing, beam alignment, or beam synthesis. For example, the authors of \cite{Heng2022SiteSpecificProbing} studied site-specific probing and later developed a grid-free beam-alignment framework in \cite{Heng2024GridFree} based on learned site-dependent measurements. The authors of \cite{Ning2023RSRPCodebook} proposed an RSRP-based environment-adaptive codebook design , while the authors of \cite{Wu2024CKMBeamforming} investigated environment-aware hybrid beamforming by leveraging channel knowledge maps. More recently, generative site-specific beamforming frameworks proposed in \cite{sim,SSBFMAG2} have further shown that coarse environment-dependent observations can already be highly informative for downstream beamforming and subspace inference . However, these works usually depart from the current standardized limited-feedback architecture and therefore do not explain how site-specific designs should be systematically incorporated into a limited-feedback system, nor how the resulting gain should be quantified under the same overhead-efficiency metric as conventional feedback schemes. \\
\indent Taken together, these two lines of research reveal a missing middle ground. Existing codebook-based feedback schemes preserve practical NR signaling structures, but they still rely on online dominant-subspace discovery from high-dimensional CSI at the UE, which leads to substantial CSI acquisition, feedback, and UE-side processing overhead. By contrast, site-specific and environment-aware beamforming works suggest that wireless propagation around a BS is often not arbitrary: persistent geometric features such as dominant reflectors, blockage patterns, and street layouts induce recurring low-dimensional channel structures that can be learned from historical or simulated site-specific channel data. In this paper, this learned propagation prior is referred to as \emph{site-specific information (SSI)}. However, existing site-specific methods usually fall outside the standardized limited-feedback architecture and therefore do not explain how such prior knowledge should be systematically incorporated into a practical limited-feedback system.
\indent This observation motivates using SSI not as a replacement for CSI feedback, but as prior information to assist the feedback process itself. Specifically, we develop an \emph{SSI-enhanced limited-feedback scheme} in which the BS combines the learned site-specific prior with a low-overhead reference signal received power (RSRP) fingerprint collected during the synchronization signal block (SSB) probing stage to infer a UE-dependent dominant subspace before explicit CSI acquisition. The UE then only needs to estimate and feed back the low-dimensional effective CSI coefficients within that inferred subspace. In this way, the proposed design preserves the practical limited-feedback structure while reducing online CSI acquisition overhead and UE-side complexity. The main contributions are summarized as follows:
\begin{itemize}
	\item We develop a unified optimization framework for limited-feedback beamforming, where CSI acquisition, UE-side compression, BS-side reconstruction, and the resulting spectral efficiency are jointly characterized from a subspace-projection perspective. This framework also unifies the interpretation of standardized NR and the proposed schemes, enabling a theoretical analysis of their CSI-capture efficiency, overhead, and computational complexity tradeoffs.
	\item We propose a site-specific Type-II feedback scheme, where the BS leverages offline learned SSI and a low-overhead RSRP fingerprint from SSB probing to infer a UE-dependent dominant beam subspace before explicit CSI acquisition. Based on this subspace, the UE only needs to estimate and feed back low-dimensional effective channel coefficients, thus retaining the representation advantage of Type-II while substantially reducing online CSI acquisition overhead.
	\item We propose a learning-based end-to-end method for the joint design of the probing codebook and the BS-side subspace inference network. In particular, the two components are co-optimized to ensure that the probing stage generates highly informative RSRP fingerprints while the inference stage delivers accurate dominant-subspace predictions, thereby maximizing CSI-capture efficiency under limited overhead.
	\item Extensive simulations show that the proposed framework achieves a significantly better overhead-efficiency tradeoff than conventional feedback schemes. In particular, the learned probing design consistently outperforms competing baselines, while the proposed site-specific Type-II scheme attains CSI-capture capability on par with Type-II with substantially lower online CSI acquisition overhead and UE-side complexity.
\end{itemize}

\indent The remainder of this paper is organized as follows. Section~\ref{sec:model} introduces the system model, including the geometric channel model and the SSB-based initial access procedure. Section~\ref{sec:framework} develops the unified optimization framework for limited-feedback beamforming. Section~\ref{sec:conventional} reinterprets the conventional Type-I, Type-II, and PSC mechanisms within this framework. Section~\ref{sec:proposed} presents the proposed site-specific Type-II feedback scheme and its theoretical analysis. Section~\ref{sec:realization} formulates the joint design of the site-specific probing codebook and the subspace inference rule and develops the corresponding learning-based solver. Section~\ref{sec:simulation} provides numerical results. Finally, Section~\ref{sec:conclusion} concludes the paper. \\
\indent \textit{Notation:} Scalars, vectors, matrices, and sets are denoted by italic letters, boldface lowercase letters, boldface uppercase letters, and calligraphic letters, respectively. $({\cdot})^T$, $({\cdot})^H$, and $({\cdot})^{\dagger}$ denote the transpose, Hermitian transpose, and Moore--Penrose pseudoinverse, respectively. $\|\cdot\|$ and $|\cdot|$ denote the Euclidean norm and the absolute value, respectively. $\mathrm{span}(\cdot)$ denotes the subspace spanned by its argument, $\mathbf{I}_N$ denotes the $N\times N$ identity matrix, and $\mathbb{E}[\cdot]$ denotes the expectation operator. $\mathcal{CN}(\mu,\mathbf{\Sigma})$ denotes the circularly symmetric complex Gaussian (CSCG) distribution with mean $\mu$ and covariance $\mathbf{\Sigma}$.

\section{System Model} \label{sec:model}
We consider a single-cell DL system in which a BS equipped with $N_t$ fully digital transmit antennas serves a single-antenna UE at a time. The UE is located on a two-dimensional (2D) plane, and the DL channel is assumed to be block fading, i.e., approximately constant within one coherence interval. The beam management and feedback procedures considered in this paper are carried out within such a coherence interval. 

\subsection{Channel Model}
We adopt a sparse geometric channel model, which represents the channel as a superposition of $L$ dominant propagation paths \cite{Ayach2014SpatiallySparse}. Let $\mathbf{h}\in\mathbb{C}^{N_t\times 1}$ denote the instantaneous DL channel vector. It is represented as
\begin{equation} \label{h}
	\mathbf{h} = \sum_{l=1}^{L} \alpha_l \mathbf{a}(\varphi_l) = \mathbf{A} \boldsymbol{\alpha},
\end{equation}
where $\alpha_l \in \mathbb{C}$ denotes the complex gain of the $l$-th path, including large-scale attenuation and small-scale fading, and $\mathbf{a}(\varphi_l)$ is the transmit-array steering vector associated with the angle of departure $\varphi_l$. We collect the steering vectors and path gains as $\mathbf{A}=[\mathbf{a}(\varphi_1),\ldots,\mathbf{a}(\varphi_L)]\in\mathbb{C}^{N_t\times L}$ and $\boldsymbol{\alpha}=[\alpha_1,\ldots,\alpha_L]^T\in\mathbb{C}^{L\times 1}$.
\begin{assumption}[\emph{Near-orthogonal angular model}] \label{ass1} \normalfont
	The BS employs a uniform linear array (ULA) with inter-element spacing $d$, and the UE lies in the far field of the array. For the $l$-th path, define the spatial frequency as $u_l\triangleq \frac{d}{\lambda}\sin(\varphi_l)$, so that the normalized steering vector is
	\begin{equation}
		\mathbf{a}(u_l)=\frac{1}{\sqrt{N_t}} \left[1, e^{j2\pi u_l}, \cdots, e^{j2\pi (N_t-1)u_l} \right]^T.
	\end{equation}
	The dominant path spatial frequencies are sufficiently separated, and the array aperture is large enough to resolve them. Consequently, the off-diagonal correlations among the dominant steering vectors are small, and the array response matrix satisfies the approximate orthogonality condition
	\begin{equation}
		\mathbf{A}^H\mathbf{A} \approx \mathbf{I}_L.
	\end{equation}
\end{assumption}
\indent \textbf{Assumption~\ref{ass1}} is the channel-structure condition used in the subsequent path-capture analysis. It implies that, in sparse site-specific environments, the channel energy can be approximately decomposed across a small number of resolvable angular components, which makes a subspace-capture interpretation meaningful. With a unit-power transmit beamformer $\mathbf{w}\in\mathbb{C}^{N_t\times 1}$ and data symbol $s$, the UE receives
\begin{equation} \label{pdsch}
	y = \sqrt{P_t} \mathbf{h}^H\mathbf{w}s + n,
\end{equation}
where $P_t$ is the transmit power and $n\sim\mathcal{CN}(0,\sigma_n^2)$ is temporally independent additive noise.
\vspace{-0.3cm}
\subsection{SSB-Based Initial Access and RSRP Fingerprint}
Before CSI feedback, the UE first performs initial access through SSB beam sweeping. Specifically, the BS transmits a predefined SSB probing codebook $\mathbf{B}=[\mathbf b_1,\ldots,\mathbf b_K]\in\mathbb{C}^{N_t\times K}$ \cite{3GPP38211,3GPP38215}. For the $i$-th probing beam, the received SSB signal is modeled as
\begin{equation}
	{{\mathbf{y}}_{{\rm{SSB}}, i}} = \sqrt{P_{\text{SSB}}}\mathbf{h}^H\mathbf{b}_i{{\bf{s}}_{{\rm{SSB}}}} + {\bf{n}}_{{\rm{SSB}}, i},
\end{equation}
where $P_{\text{SSB}}$ is the SSB transmit power, $\mathbf{s}_{\rm SSB}\in\mathbb{C}^{L_s\times 1}$ collects the SSB symbols used for this measurement, and $\mathbf{n}_{{\rm SSB},i}\sim\mathcal{CN}(\mathbf 0_{L_s},\sigma_n^2\mathbf I_{L_s})$ is the corresponding noise vector \footnote{In NR, SSB measurements are obtained over specific time-frequency resource elements. For analytical compactness, we aggregate these resource elements into a vector and do not separately track their time-frequency indices.}. The UE measures the RSRP of the $i$-th beam by averaging the received SSB power over the aggregated SSB symbols,
\begin{equation}
	\text{RSRP}_{\mathbf{b}_i} \triangleq \frac{1}{L_s}\sum_{t=1}^{L_s} \left|y_{{\rm SSB},i}^{(t)}\right|^2.
\end{equation}
Following \cite{sim}, the dB-domain RSRP can be modeled as a noisy observation of an ideal beam power,
\begin{equation}
	r_{\mathbf{b}_i} = 10 \log_{10} \text{RSRP}_{\mathbf{b}_i} = r_{\mathbf{b}_i}^0 + n_{\mathbf{b}_i},
\end{equation}
where $r_{\mathbf{b}_i}^0$ is the noise-free dB-domain RSRP value and $n_{\mathbf{b}_i}$ is a Gaussian perturbation with mean $\mu_b$ and variance $\sigma_b^2$. The detailed expressions of $r_{\mathbf{b}_i}^0$, $\mu_b$, and $\sigma_b^2$ can be found in Eq. (7), Eq. (8a), and Eq. (8b) of \cite{sim}. After sweeping all $K$ SSB beams, the UE obtains the RSRP fingerprint
\begin{equation} \label{RSRP}
	\mathbf{r}_\mathbf{B} = \mathbf{r}_\mathbf{B}^0 + \mathbf{n}_\mathbf{B},
\end{equation}
where $\mathbf{r}_\mathbf{B}^0=[r_{\mathbf{b}_1}^0,\ldots,r_{\mathbf{b}_K}^0]^T$ and $\mathbf{n}_\mathbf{B}=[n_{\mathbf{b}_1},\ldots,n_{\mathbf{b}_K}]^T$. During initial access, the UE measures $\mathbf{r}_\mathbf{B}$ to determine the access beam and the associated random-access resources before subsequent CSI acquisition. \\
\indent The above initial-access stage only provides coarse beam-level information for synchronization and access establishment, but it is generally insufficient for high-quality downlink beamforming. To support data transmission, the BS still needs a more refined UE-specific CSI representation, since the resulting beamforming performance critically depends on the quality of CSI acquisition. In practical limited-feedback systems, this requires additional CSI-reference signal (CSI-RS) transmission, UE-side channel estimation, CSI compression, and feedback. When the transmit-antenna dimension is large, such CSI acquisition can incur substantial training, feedback, and processing overhead. Therefore, the key issue is how efficiently this representation can be acquired under limited online resources.
\vspace{-0.3cm}
\section{A Unified Framework for Limited-Feedback Beamforming} \label{sec:framework}
\subsection{CSI Acquisition and Feedback Pipeline}
Recall the DL data transmission model in (\ref{pdsch}). If the DL CSI $\mathbf h$ is available at the BS, the single-user single-stream beamforming problem can be characterized as
\begin{equation} \label{eq:ideal_bf}
	\mathbf{w}^\star = \arg\max_{\mathbf{w}: \|\mathbf{w}\|^2=1} |\mathbf{h}^H\mathbf{w}|^2,
\end{equation}
with the optimal solution to be the matched-filter or maximum-ratio-transmission (MRT) beamformer $\mathbf{w}^\star=\mathbf h/\|\mathbf h\|$. In practice, however, the BS does not directly know the UE-specific DL channel. In TDD systems, uplink sounding and channel reciprocity may provide a useful estimate \cite{Larsson2014MassiveMIMO}, but calibration and refinement can still be required when reciprocity is imperfect or when the full-dimensional channel is not directly observable. In FDD systems, the DL CSI must instead be estimated at the UE and fed back to the BS. \\
\indent We abstract this DL acquisition process as three coupled stages: 1) BS-side CSI-RS transmission, 2) UE-side CSI extraction and compression, and 3) BS-side reconstruction. Let $\mathbf{S}_{\text{CSI}}\in\mathbb{C}^{N_c\times L_c}$ denote the CSI-RS training matrix transmitted over $L_c$ channel uses from $N_c$ CSI-RS ports, which represent effective pilot transmission dimension seen by the UE, rather than necessarily a physical transmit antenna, and let $\mathbf C\in\mathbb{C}^{N_t\times N_c}$ denote the CSI-RS precoder \footnote{As for SSB, CSI-RS symbols over time-frequency resource elements are aggregated into a training matrix for analytical compactness.}. The received CSI-RS is given by
\begin{equation}
	\mathbf{y}_{\text{CSI}}^T(\mathbf{C}) = \sqrt{P_{\text{CSI}}}\,\mathbf{h}^H \mathbf{C} \mathbf{S}_{\text{CSI}} + \mathbf{n}_{\text{CSI}}^T,
\end{equation}
where $P_{\text{CSI}}$ is the CSI-RS transmit power and $\mathbf{n}_{\text{CSI}}\sim\mathcal{CN}(\mathbf 0_{L_c},\sigma_n^2\mathbf I_{L_c})$. Based on $\mathbf{y}_{\text{CSI}}(\mathbf C)$, the UE obtains a CSI representation through an acquisition mapping $\mathcal R$, e.g., least-squares (LS) or minimum mean square error (MMSE) estimation \cite{Hassibi2003Training}. In general, $L_c$ should be no less than $N_c$ to reliably recover $\mathbf{h}$. Hence, the overhead for DL channel estimation is at least $\mathcal{O}(N_c^2)$. To avoid distraction from estimation error, we assume perfect acquisition under the required training dimension and write $\mathbf h=\mathcal R(\mathbf{y}_{\text{CSI}})$ with $\mathcal{R}: \mathbb{C}^{L_c \times 1} \rightarrow \mathbb{C}^{N_t \times 1}$. \\
\indent Since feeding back full CSI is generally prohibitive, the UE compresses the acquired CSI into a lower-dimensional feedback message
\begin{equation} \label{eq:q_mapping}
	\mathbf{z} = \mathcal{Q}(\mathbf{h}),
\end{equation}
where $\mathcal Q:\mathbb{C}^{N_t\times 1}\rightarrow\mathbb{C}^{N_q\times 1}$ represents the UE-side compression rule. The dimension $N_q$ should be understood as a normalized feedback-payload proxy, e.g., the number of scalar coefficients or indices conveyed to the BS. Here, $N_q$ is only adopted as an abstract channel-use-equivalent feedback overhead. The detailed modeling of the uplink feedback channel is beyond the scope of this paper. After receiving $\mathbf z$, the BS reconstructs the channel as
\begin{equation} \label{eq:f_mapping}
	\hat{\mathbf{h}} = \mathcal{F}(\mathbf{z}),
\end{equation}
where $\mathcal F:\mathbb{C}^{N_q\times 1}\rightarrow\mathbb{C}^{N_t\times 1}$ is the BS-side reconstruction mapping. In general, $\mathcal F$ is not the inverse of $\mathcal Q$, because the UE-side compression is lossy. As a consequence, the optimal MRT beamformer based on this reconstructed channel is given by
\begin{equation} \label{eq:unified_pipeline}
	\hat{\mathbf{w}} = \frac{\hat{\mathbf{h}}}{\| \hat{\mathbf{h}} \|}  = 	\frac{\mathcal{F}\big(\mathcal{Q}(\mathcal{R}(\mathbf{y}_{\text{CSI}} (\mathbf{C})))\big)}{\|\mathcal{F}\big(\mathcal{Q}(\mathcal{R}(\mathbf{y}_{\text{CSI}} (\mathbf{C})))\big)\|}.
\end{equation}
\vspace{-0.6cm}
\subsection{Beamforming Design with Subspace Projection}
The above pipeline makes the overhead-performance tradeoff explicit. Specifically, the online overhead consists of SSB-stage overhead and CSI acquisition/feedback overhead, which can be calculated as $T_o=T_{\rm SSB}+T_{\rm CSI}$ in channel uses. Here, $T_{\rm SSB}$ accounts for SSB sweeping and, when needed, the associated RSRP reporting, while $T_{\rm CSI}$ accounts for CSI-RS transmission and the subsequent CSI feedback payload. For a coherence block of $T_c$ channel uses and denoting $\rho=P_t/\sigma_n^2$, the limited-feedback beamforming design problem for maximizing the effective rate can be abstracted as
\begin{equation} \label{eq:unified_problem}
	\max_{\mathbf{C},\,\mathcal{Q},\,\mathcal{F}}\; \left( {1 - \frac{{{T_o}}}{{{T_c}}}} \right){\log _2}\left( {1 + {\rho}{{\left| {{{\mathbf{h}}^H}\hat{\mathbf{w}}} \right|}^2}} \right).
\end{equation}
\indent To streamline the beamforming design, we now introduce a subspace project theory of the mapping given in (13). In particular, $\mathcal Q$ and $\mathcal F$ determine the UE-feedback information and the CSI reconstruction quality, which essentially inducing a scheme-specific CSI representation subspace $\mathcal U$ together. Let $\mathbf{U}$ be a basis matrix such that $\mathcal{U}=\mathrm{span}(\mathbf{U})$. The orthogonal projector onto $\mathcal U$ is then given by $\mathbf{P}_\mathcal{U}=\mathbf{U}(\mathbf{U}^H\mathbf{U})^{-1}\mathbf{U}^H$.
\begin{lemma}[\emph{Subspace-projection characterization}] \label{theo1} \normalfont
For any limited-feedback scheme, i.e., $\mathcal Q$ and $\mathcal F$, whose reconstructed CSI is restricted to a subspace $\mathcal{U}$, the optimal MRT beamformer supported by this subspace is
\begin{equation}
	\hat{\mathbf{w}}_\mathcal{U}=\frac{\mathbf{P}_\mathcal{U}\mathbf{h}}{\|\mathbf{P}_\mathcal{U}\mathbf{h}\|},
\end{equation}
where $\mathbf{P}_\mathcal{U}\mathbf{h}\neq \mathbf{0}$. The resulting beamforming gain is $\|\mathbf{P}_\mathcal{U}\mathbf{h}\|^2$. Under \textbf{Assumption~\ref{ass1}}, it can be approximated as
\begin{equation}
	\|\mathbf{P}_\mathcal{U}\mathbf{h}\|^2 \approx \sum_{l=1}^{L}|\alpha_l|^2 g_l,
\end{equation}
where $g_l\triangleq\|\mathbf{P}_\mathcal{U}\mathbf{a}(u_l)\|^2\in[0,1]$ is the projected channel gain of the $l$-th path. 
\end{lemma}
\vspace{-0.3cm}
\begin{IEEEproof}
For a beamformer restricted to $\mathcal U$, the MRT direction is the normalized projection of $\mathbf h$ onto $\mathcal U$, which yields $\hat{\mathbf{w}}_\mathcal{U}=\mathbf{P}_\mathcal{U}\mathbf h/\|\mathbf{P}_\mathcal{U}\mathbf h\|$. Thus,
\begin{align}
	|\mathbf h^H\hat{\mathbf w}_\mathcal{U}|^2
	&=\left|\frac{\mathbf h^H\mathbf P_\mathcal U\mathbf h}{\|\mathbf P_\mathcal U\mathbf h\|}\right|^2
	=\|\mathbf P_\mathcal U\mathbf h\|^2 \\
	&=\boldsymbol{\alpha}^H\mathbf A^H\mathbf P_\mathcal U\mathbf A\boldsymbol{\alpha}.
\end{align}
Under \textbf{Assumption~\ref{ass1}}, the cross-path correlations are weak and the off-diagonal terms in $\mathbf A^H\mathbf P_\mathcal U\mathbf A$ are negligible, giving $\|\mathbf{P}_\mathcal{U}\mathbf{h}\|^2 \approx \sum_{l=1}^{L}|\alpha_l|^2\|\mathbf{P}_\mathcal{U}\mathbf{a}(u_l)\|^2$. Dividing this expression by $\|\mathbf h\|^2\approx\sum_{l=1}^{L}|\alpha_l|^2$ completes the proof.
\end{IEEEproof}
\vspace{-0.2cm}
\begin{theorem}[\emph{CSI-capture efficiency}] \label{efficiency} \normalfont
	Under \textbf{Lemma~\ref{theo1}} and \textbf{Assumption~\ref{ass1}}, the CSI-capture efficiency satisfies
	\begin{equation}
		\eta \triangleq \frac{\|\mathbf{P}_\mathcal{U}\mathbf{h}\|^2}{\|\mathbf{h}\|^2}
		\approx
		\frac{\sum_{l=1}^{L}|\alpha_l|^2 g_l}{\sum_{l=1}^{L}|\alpha_l|^2},
	\end{equation}
	where $g_l\triangleq\|\mathbf{P}_\mathcal{U}\mathbf{a}(u_l)\|^2\in[0,1]$ denotes the projected gain of the $l$-th path.
\end{theorem}
\vspace{-0.3cm}
\begin{IEEEproof}
	The result follows directly from \textbf{Lemma~\ref{theo1}}.
\end{IEEEproof}
\textbf{Lemma~\ref{theo1}} and \textbf{Theorem~\ref{efficiency}} show that the performance of a limited-feedback scheme is governed by how efficiently its induced representation subspace captures channel energy. The path-capture factor $g_l$ measures the fraction of the $l$-th path energy preserved by the scheme, and the overall CSI-capture efficiency is a path-power-weighted average of these factors. Therefore, problem (\ref{eq:unified_problem}) can be approximately rewritten as
\begin{equation} \label{eq:unified_problem2}
	\max_{\mathbf{C},\,\mathcal{Q},\,\mathcal{F}}\; \left( {1 - \frac{{{T_o}}}{{{T_c}}}} \right){\log _2}\left(1+\rho \sum_{l=1}^{L} |\alpha_l|^2 g_{l}\right).
\end{equation}
Problem (\ref{eq:unified_problem}) and (\ref{eq:unified_problem2}) are both difficult to solve directly because the CSI acquisition, compression, and reconstruction mappings are coupled. Nevertheless, the subspace-projection viewpoint reveals a useful design principle: \emph{a limited-feedback scheme should preserve as much dominant channel energy as possible while minimizing CSI-RS and feedback overhead.} The following sections use this viewpoint to reinterpret standardized mechanisms and to motivate the proposed SSI-enhanced design.
\vspace{-0.1cm}
\section{Conventional Limited-Feedback Schemes} \label{sec:conventional}
Current 3GPP NR standards do not solve problem (\ref{eq:unified_problem}) in a direct optimization sense. Instead, they approximate this design objective through standardized limited-feedback schemes, including Type-I, Type-II, and PSC \cite{3GPP38214,Fu2023TutorialCodebooks}. Within the unified framework developed above, these methods define practical implementations of the UE-side compression mapping $\mathcal{Q}$ and the BS-side reconstruction mapping $\mathcal{F}$ under explicit signaling and complexity constraints. Their main differences lie in the representation granularity of CSI, the resulting feedback payload, and hence, the beamforming overhead-performance tradeoff. In this section, we reinterpret these mechanisms through the unified optimization viewpoint established in Section \ref{sec:framework}.
\vspace{-0.3cm}
\subsection{Type-I Feedback}
The Type-I feedback serves as the simplest and lightest-weight limited-feedback beamforming mechanism in current NR systems, offering low implementation complexity and feedback overhead at the cost of relatively low spectral efficiency. 
\subsubsection{Characterization using the Proposed Framework}
Within the unified formulation in problem (\ref{eq:unified_problem}), the Type-I feedback can be viewed as an extremely coarse CSI-quantization scheme. Concretely, since the BS does not know the UE-specific dominant directions a priori, whereas the UE requires sufficiently comprehensive channel observations to perform reliable codebook-based beam selection, non-precoded CSI-RS, i.e., $\mathbf{C}_{\rm I} = \mathbf{I}_{N_t}$, is first adopted to probe all effective transmit ports in a direction-agnostic manner \footnote{Typically, each port is linked to multiple physical antennas, allowing for enhanced signal quality and increased capacity in the communication system. However, we assume the number of CSI-RS ports equals the number of transmit antennas here for analytical simplicity.}. \\
\indent After receiving CSI-RS, the UE obtains the perfect DL CSI $\mathbf{h}$ as assumed previously. To reduce CSI feedback overhead, the BS and UE share a predefined structured codebook $\mathbf{D}$ to quantize the channel, which is typically an oversampled discrete Fourier transform (DFT) codebook. For a ULA with $N_t$ antennas, the quantization codebook is of dimension $N_t \times N_t O_D$ with oversampling factor $O_D$. The $m$-th column of $\mathbf{D}$ can be constructed as
\begin{equation}
\mathbf{d}_m = \frac{1}{\sqrt{N_t}} \left[1, e^{j 2\pi\frac{ m}{N_tO_D}}, \ldots, e^{j2\pi \frac{ m (N_t-1)}{N_tO_D}} \right]^T,
\end{equation}
where $m=0,1,\cdots, N_tO_D-1$. In the Type-I feedback, the UE selects the beam that best matches the estimated channel from $\mathbf{D}$. The corresponding index is then fed back as the precoding matrix indicator (PMI) to the BS \footnote{Other CSI quantities such as the rank indicator and channel quality indicator are also fed back in this process to support multi-layer transmission and modulation-and-coding selection. However, these quantities are beyond the scope of this paper and are not considered further.}. In this manner, the UE-side CSI compression mapping can be expressed as
\begin{equation}
	z_{\rm I} = \mathcal{Q}_{\rm I}(\mathbf{h})= \argmax_{m=1,\cdots,N_tO_D} \left|\mathbf{h}^H\mathbf{d}_m\right|^2.
\end{equation}
Note that quantization of $z_{\rm I}$ in the transmission process is not considered in this paper. \\
\indent After receiving $z_{\rm I}$, the BS reconstructs the CSI from the PMI using the predefined oversampled DFT codebook $\mathbf{D}$ as
\begin{equation}
	\hat{\mathbf{h}}_{\rm I} = \mathcal{F}_{\rm I}(z_{\rm I}) = \mathbf{d}_{z_{\rm I}}.
\end{equation}
Type-I feedback then defines a one-dimensional subspace $\mathcal U_{\rm I}=\mathrm{span}(\mathbf d_{z_{\rm I}})$ with projector $\mathbf{P}_{\rm I}=\mathbf{d}_{z_{\rm I}}\mathbf{d}_{z_{\rm I}}^H$. As a consequence, the optimal Type-I beamforming vector is given by
\begin{equation}
	\hat{\mathbf{w}}_{\rm I} = \frac{\mathcal{F}_{\rm I}\big(\mathcal{Q}_{\rm I}(\mathcal{R}(\mathbf{y}_{\text{CSI}} (\mathbf{I}_{N_t})))\big)}{\|\mathcal{F}_{\rm I}\big(\mathcal{Q}_{\rm I}(\mathcal{R}(\mathbf{y}_{\text{CSI}} (\mathbf{I}_{N_t})))\big)\|}=\mathbf{d}_{z_{\rm I}},
\end{equation}
This expression shows that Type-I scheme realizes limited-feedback beamforming by quantizing the spatial channel with a DFT basis and selecting a single direction that best matches the channel. Hence, the Type-I scheme primarily beams toward one dominant direction on the predefined grid. This makes Type-I feedback a low-complexity and low-overhead feasible realization of (\ref{eq:unified_problem}), but its achievable rate is constrained by the codebook resolution and the rank-one representation. The resulting beamforming loss can be pronounced when the actual channel direction is mismatched with the predefined grid, or when the channel energy is spread over multiple propagation paths. Type-I feedback therefore serves as a baseline beam-management mechanism in current NR systems and is particularly suitable for scenarios with extremely limited feedback capability.

\subsubsection{CSI-Capture Efficiency Analysis}
The Type-I feedback is equivalent to choosing the projector that captures the largest channel energy. Accordingly, the achievable rate for Type-I feedback is given by
\begin{align}
R_{\rm I} &= \log_2\left(1+\rho\left| \mathbf{h}^H\hat{\mathbf{w}}_{\rm I}\right|^2\right) = \log_2\left(1+\rho\|\mathbf{P}_{\rm I}\mathbf{h}\|^2\right) \notag \\ &= {\log _2}\left( {1 + \rho {{\left| {{{\bf{h}}^H}{\bf{d}}_{z_{\rm I}}} \right|}^2}} \right) = {\log _2}\left( {1 + \rho {{\left| {{\boldsymbol{\alpha} ^H}{{\bf{A}}^H}{\bf{d}}_{z_{\rm I}}} \right|}^2}} \right) \notag \\ &\overset{(a)}{\approx} {\log _2}\left( {1 + \rho {{\left| \alpha_{l_{\rm I}} \right|}^2}{{\left| {{{\bf{a}}^H}\left( {u}_{l_{\rm I}} \right){\bf{d}}_{z_{\rm I}}} \right|}^2}} \right),
\end{align}
where $l_{\rm I} \triangleq \arg\max_l |\alpha_l|^2$ denotes the strongest path. The approximation in step $(a)$ is based on \textbf{Assumption~\ref{ass1}}. More particularly, for Type-I feedback, it has $g_l \approx 0, \forall l \neq l_{\rm I}$ in (\ref{eq:unified_problem2}) and 
\begin{align}
	g_{l_{\rm I}} &= {{\left| {{{\bf{a}}^H}\left( {u}_{l_{\rm I}} \right){\bf{d}}_{z_{\rm I}}} \right|}^2} = {\left| \frac{1}{N_t} \sum_{n=0}^{N_t-1} e^{j2\pi n\left(u_{l_{\rm I}}-\frac{{z_{\rm I}}}{N_tO_D}\right)} \right|}^2 \notag \\ &= \frac{1}{N_t^2} \left|\frac{\sin\!\left(\pi N_t \Delta_{l,m}\right)} {\sin\!\left(\pi \Delta_{l,m}\right)} \right|^2,
\end{align}
where $\Delta_{{l_{\rm I}},{z_{\rm I}}} = u_{l_{\rm I}} - \frac{{z_{\rm I}}}{N_tO_D}$. For the oversampled DFT codebook, we have $|\Delta_{{l_{\rm I}},{z_{\rm I}}}| \le \frac{1}{2N_tO_D}$, which yields $g_{l_{\rm I}} \in [g_{\min},1]$ with 
\begin{equation}
	g_{\min} = \frac{1}{N_t^2} \left|\frac{\sin\!\left(\frac{\pi}{2O_D}\right)}{\sin\!\left(\frac{\pi}{2N_tO_D}\right)} \right|^2.
\end{equation} 
Overall, the spectral efficiency of Type-I feedback can be approximated as
\begin{equation}
	R_{\rm I} \approx {\log _2}\left( {1 + \rho {{\left| \alpha_{l_{\rm I}} \right|}^2}g_{l_{\rm I}}} \right).
\end{equation}
\indent Based on the above analysis and the results in \textbf{Theorem \ref{efficiency}}, the CSI-capture efficiency for Type-I feedback method is given in the following corollary.
\begin{corollary}[\emph{Type-I CSI-capture efficiency}] \label{lemma1} \normalfont
Under \textbf{Assumption~\ref{ass1}}, the CSI-capture efficiency of Type-I feedback can be approximated as
\begin{equation}
\eta_{\rm I} = \frac{\|\mathbf{P}_{\rm I} \mathbf{h}\|^2}{\|\mathbf{h}\|^2} \approx \frac{|\alpha_{l_{\rm I}}|^2 g_{l_{\rm I}}}{\sum_{l=1}^{L}|\alpha_l|^2}.
\end{equation}
\end{corollary}
\textbf{Corollary~\ref{lemma1}} shows that the loss of Type-I feedback is caused by two factors: the energy outside the single selected dominant path and the DFT-grid mismatch captured by $g_{l_{\rm I}}$. The first factor limits Type-I feedback in multi-path channels, while the second motivates richer or learned representation bases beyond single-beam quantization.
\vspace{-0.3cm}
\subsection{Type-II Feedback}
To overcome the coarse spatial quantization of Type-I feedback, Type-II scheme provides a more flexible and higher-resolution beam-domain representation of the channel. 
\subsubsection{Characterization using the Proposed Framework}
As in Type-I feedback, non-precoded CSI-RS is still adopted so that the UE can probe the full-port DL CSI space, i.e., $\mathbf{C}_{\rm II} = \mathbf{I}_{N_t}$. After acquiring the perfect DL CSI, the UE no longer selects only a single beam. Instead, it reports a small set of beams from a predefined oversampled DFT dictionary together with their combining coefficients. Consequently, Type-II feedback retains a CSI-RS overhead on the same order as Type-I, while incurring a higher feedback cost because both beam indices and coefficients must be conveyed. \\
\indent Concretely, the UE selects $Q$ dominant spatial directions from the oversampled DFT codebook $\mathbf{D}$ based on the DL CSI. Supposing $\mathcal S \subseteq \{0,\ldots,N_tO_D-1\}$ and $|\mathcal S|=Q$, the corresponding beam subspace is constructed as $\mathcal{U}_{\mathcal S} = \mathrm{span}(\mathbf{D}_{\mathcal S}) = \mathrm{span}([\mathbf{d}_m]_{m\in\mathcal S})$. Therefore, the UE determines the optimal beam indices by solving
\begin{equation}
\mathcal S_{\rm II} = \arg\max_{\mathcal S} \left\| \mathbf P_{\mathcal{U}_{\mathcal S}} \mathbf h \right\|^2,
\end{equation}
where $\mathbf P_{\mathcal{U}_{\mathcal S}} = \mathbf{D}_\mathcal S(\mathbf{D}_\mathcal S^H \mathbf{D}_\mathcal S)^{-1}\mathbf{D}_\mathcal S^H$ is the orthogonal projector onto the subspace spanned by the selected beams. The CSI representation subspace defined by Type-II feedback is therefore $\mathcal{U}_{\rm II} = \mathrm{span}(\mathbf{D}_{\mathcal S_{\rm II}})$. The UE then estimates the combining coefficients by minimizing the compression error
\begin{equation}
	\boldsymbol{\alpha}_{\rm II} = \arg\min_{\boldsymbol{\alpha}} \|\mathbf{h} - \mathbf{D}_{\mathcal S_{\rm II}} \boldsymbol{\alpha}\|^2=\mathbf{D}_{\mathcal S_{\rm II}}^{\dagger}\mathbf{h}.
\end{equation}
Therefore, the UE-side compression mapping can be expressed as
\begin{equation}
\mathbf{z}_{\rm II} = \mathcal{Q}_{\rm II}(\mathbf{h}) = \left(\mathcal{S}_{\rm II}, \boldsymbol{\alpha}_{\rm II}\right).
\end{equation}
More particularly, the UE reports both the selected beam indices and the corresponding combining coefficients to the BS. After receiving $\mathbf{z}_{\rm II}$, the BS reconstructs the DL CSI via
\begin{equation}
\hat{\mathbf{h}}_{\rm II} = \mathcal{F}_{\rm II}(\mathbf{z}_{\rm II}) = \mathbf{D}_{\mathcal S_{\rm II}}\boldsymbol{\alpha}_{\rm II} = \mathbf{D}_{\mathcal S_{\rm II}}\mathbf{D}_{\mathcal S_{\rm II}}^{\dagger}\mathbf{h} = \mathbf{P}_{\rm II}\mathbf{h},
\end{equation}
which is the projection of the actual channel onto the subspace spanned by the selected directions. The resulting Type-II beamformer is
\begin{equation}
\hat{\mathbf{w}}_{\rm II} = \frac{\mathcal{F}_{\rm II}\big(\mathcal{Q}_{\rm II}(\mathcal{R}(\mathbf{y}_{\text{CSI}} (\mathbf{I}_{N_t})))\big)}{\|\mathcal{F}_{\rm II}\big(\mathcal{Q}_{\rm II}(\mathcal{R}(\mathbf{y}_{\text{CSI}} (\mathbf{I}_{N_t})))\big)\|}=\frac{\mathbf{P}_{\rm II}\mathbf{h}}{\|\mathbf{P}_{\rm II}\mathbf{h}\|},
\end{equation}
Compared with Type-I feedback, the feasible beamforming set is no longer restricted to a single codeword but to the subspace spanned by a set of selected beams. This provides greater representation flexibility and usually improves the achievable beamforming gain in multipath channels. However, the UE must determine both the dominant beam directions and the corresponding coefficients based on the estimated high-dimensional channel. Both beam selection and Moore--Penrose pseudoinverse calculation impose a high computational burden on the UE and require a larger feedback payload.

\subsubsection{CSI-Capture Efficiency Analysis}
The achievable rate for Type-II feedback is given by
\begin{align}
	R_{\rm II} &= \log_2\left(1+\rho\left| \mathbf{h}^H\hat{\mathbf{w}}_{\rm II}\right|^2\right) = \log_2\left(1+\rho\left|\frac{\mathbf{h}^H\mathbf{P}_{\rm II}\mathbf{h}} {\|\mathbf{P}_{\rm II}\mathbf{h}\|}\right|^2\right) \notag\\ &= \log_2\left(1+\rho\|\mathbf{P}_{\rm II}\mathbf{h}\|^2\right) = \log_2\left(1+\rho\,\boldsymbol{\alpha}^H\mathbf{A}^H\mathbf{P}_{\rm II}\mathbf{A}\boldsymbol{\alpha}\right) \notag\\ &\overset{(b)}{\approx} \log_2\left(1+\rho\sum_{l\in\mathcal L_{Q,\rm II}}|\alpha_l|^2 g_{l,\rm II}\right),
\end{align}
where $\mathcal L_{Q,\rm II}$ denotes the set of dominant paths captured by the selected Type-II subspace. The approximation in (b) is again based on \textbf{Assumption~\ref{ass1}}. Thus, the path-capture factor satisfies $g_{l,\rm II}\approx 0$ for $l\notin\mathcal L_{Q,\rm II}$, whereas for captured paths it is given by
\begin{align}
	g_{l,\rm II} &\triangleq \|\mathbf{P}_{\rm II}\mathbf{a}(u_l)\|^2 \approx \sum_{m\in\mathcal S_{\rm II}} \left|\mathbf{a}^H(u_l)\mathbf{d}_m\right|^2 \notag\\ &= \sum_{m\in\mathcal S_{\rm II}} \frac{1}{N_t^2} \left| \frac{\sin\!\left(\pi N_t\Delta_{l,m}\right)} {\sin\!\left(\pi \Delta_{l,m}\right)} \right|^2.
\end{align}
For additional intuition, we further have
\begin{equation}
	g_{l,\rm II} \gtrsim \left|\mathbf{a}^H(u_l)\mathbf{d}_{m_l^\star}\right|^2 = g_{l_{\rm I}},
\end{equation}
where $m_l^\star=\arg\min_{m\in\mathcal S_{\rm II}}\left|u_l-\frac{m}{N_tO_D}\right|$ denotes the selected DFT atom closest to the $l$-th path. The resulting CSI-capture efficiency of the Type-II feedback is given below.
\begin{corollary}[\emph{Type-II CSI-capture efficiency}] \label{lemma2} \normalfont
Under \textbf{Assumption~\ref{ass1}}, the CSI-capture efficiency of Type-II feedback can be approximated as
\begin{equation}
\eta_{\rm II} = \frac{\|\mathbf{P}_{\rm II} \mathbf{h}\|^2}{\|\mathbf{h}\|^2} \approx \frac{\sum_{l\in\mathcal L_{Q,\rm II}}|\alpha_l|^2 g_{l,\rm II}}{\sum_{l=1}^{L}|\alpha_l|^2}.
\end{equation}
\end{corollary}
\textbf{Corollary~\ref{lemma2}} shows that Type-II feedback improves over single-beam quantization by allowing several selected dictionary beams to jointly capture multipath energy. Its remaining loss comes from the energy outside the selected subspace and from DFT-basis mismatch, while its cost is the UE-side search and coefficient-estimation burden required to construct this subspace from full-dimensional CSI.

\subsection{Port-Selection Feedback}
To reduce the high feedback overhead and UE-side computational complexity of the flexible Type-II feedback while retaining a richer channel representation than the single-beam Type-I feedback, the PSC was introduced as a more structured limited-feedback mechanism. Instead of representing the channel through a large beam-domain dictionary, PSC restricts the feedback to a selected subset of effective ports and the associated low-dimensional coefficients.
\subsubsection{Characterization using the Proposed Framework}
As in Type-I and Type-II feedback, PSC still relies on non-precoded CSI-RS to expose the full CSI-RS port space to the UE, i.e., $\mathbf{C}_{\rm PSC}=\mathbf{I}_{N_t}$. After acquiring the full-port CSI $\mathbf{h}$, the UE selects a subset of $N_p$ effective ports that preserve as much channel energy as possible. Let $\mathcal{S}\subseteq\{1,\ldots,N_t\}, |\mathcal{S}|=N_p$ denote the selected effective-port set, and let
\begin{equation}
\mathbf{E}_{\mathcal S} = [\mathbf{e}_{p}]_{p\in\mathcal{S}} \in \mathbb{C}^{N_t \times N_p}
\end{equation}
be the associated selection matrix, where $\mathbf{e}_{p}$ is the $p$-th canonical basis vector. The UE selects the effective-port subset that best captures the channel energy, which can be formulated as
\begin{equation}
\mathcal S_{\rm PSC} = \arg\max_{\mathcal S} \left\| \mathbf P_{\mathcal S} \mathbf h \right\|^2,
\end{equation}
where $\mathbf P_{\mathcal S}=\mathbf E_{\mathcal S}(\mathbf E_{\mathcal S}^H\mathbf E_{\mathcal S})^{-1}\mathbf E_{\mathcal S}^H=\mathbf E_{\mathcal S}\mathbf E_{\mathcal S}^H$ is the orthogonal projector onto the port-selected subspace. This is consistent with selecting the strongest $N_p$ effective ports. The CSI representation subspace defined by PSC is therefore $\mathcal{U}_{\rm PSC} = \mathrm{span}(\mathbf{E}_{\mathcal S_{\rm PSC}})$. \\
\indent The UE then estimates the low-dimensional coefficient vector over the selected subspace by solving the same LS problem as in Type-II feedback, which yields
\begin{equation}
\boldsymbol{\alpha}_{\rm PSC} = \arg\min_{\boldsymbol{\alpha}} \|\mathbf{h} - \mathbf{E}_{\rm PSC} \boldsymbol{\alpha}\|^2 = \mathbf{E}_{\rm PSC}^{\dagger}\mathbf{h} = \mathbf{E}_{\rm PSC}^{H}\mathbf{h}.
\end{equation}
A key distinction from Type-II feedback is that the coefficients are obtained by selecting entries from the estimated CSI rather than by performing Moore--Penrose pseudoinverse calculation over a beam-domain dictionary, which significantly reduces the UE-side computational complexity. Then, the UE reports
\begin{equation}
\mathbf{z}_{\rm PSC} = \mathcal{Q}_{\rm PSC}(\mathbf{h}) = \left(\mathcal{S}_{\rm PSC}, \boldsymbol{\alpha}_{\rm PSC}\right).
\end{equation}
Correspondingly, the BS recovers the CSI as
\begin{equation}
\hat{\mathbf{h}}_{\rm PSC} = \mathcal{F}_{\rm PSC}(\mathbf{z}_{\rm PSC}) = \mathbf{E}_{\rm PSC}\boldsymbol{\alpha}_{\rm PSC} = \mathbf{P}_{\rm PSC}\mathbf{h},
\end{equation}
which leads to the beamformer
\begin{equation}
\hat{\mathbf{w}}_{\rm PSC} = \frac{\mathcal{F}_{\rm PSC}\big(\mathcal{Q}_{\rm PSC}(\mathcal{R}(\mathbf{y}_{\text{CSI}} (\mathbf{I}_{N_t})))\big)}{\|\mathcal{F}_{\rm PSC}\big(\mathcal{Q}_{\rm PSC}(\mathcal{R}(\mathbf{y}_{\text{CSI}} (\mathbf{I}_{N_t})))\big)\|}=\frac{\mathbf{P}_{\rm PSC}\mathbf{h}}{\|\mathbf{P}_{\rm PSC}\mathbf{h}\|}.
\end{equation}
Under the unified framework, PSC offers a more structured representation, a smaller search space, and lower implementation complexity than Type-II feedback.

\subsubsection{CSI-Capture Efficiency Analysis}
Let $|\tilde h|_{(1)}^2\ge\cdots\ge|\tilde h|_{(N_c)}^2$ denote the sorted effective-port powers. The PSC projector captures the energy of the selected effective ports. Therefore, before accounting for CSI-RS and feedback overhead, the instantaneous rate associated with PSC beamforming is
\begin{align}
R_{\rm PSC} &= \log_2\left(1+\rho\left| \mathbf{h}^H\hat{\mathbf{w}}_{\rm PSC}\right|^2\right) = \log_2\left(1+\rho\|\mathbf{P}_{\rm PSC}\mathbf{h}\|^2\right) \notag \\
& = \log_2\left(1+\rho \sum_{i=1}^{N_p} |\tilde h|_{(i)}^2\right).
\end{align}
\begin{corollary}[\emph{PSC CSI-capture efficiency}] \label{lemma3} \normalfont
For a selected effective-port set of size $N_p$, the CSI-capture efficiency of PSC is
\begin{equation}
\eta_{\rm PSC} = \frac{\|\mathbf{P}_{\rm PSC}\mathbf{h}\|^2}{\|\mathbf{h}\|^2} = \frac{\sum_{i=1}^{N_p} |\tilde h|_{(i)}^2}{\|\mathbf{h}\|^2}.
\end{equation}
\end{corollary}
\textbf{Corollary~\ref{lemma3}} indicates that PSC is governed by the concentration of the effective-port-domain channel. It can be efficient when a small number of ports dominate, but it may lose substantial energy when the effective-port powers are spread more uniformly. Unlike a global ordering among Type-I, Type-II, and PSC, this interpretation is explicitly channel- and representation-dependent.

\section{Site-Specific Type-II Feedback} \label{sec:proposed}
\indent The above subspace-projection view also clarifies the main limitation of conventional feedback mechanisms. To achieve high-quality beamforming, especially in Type-II feedback, the UE must perform several costly online tasks, including high-dimensional CSI acquisition, dominant-subspace identification, and coefficient calculation within the selected subspace. Consequently, improving the subspace representation capability usually comes at the cost of substantially higher CSI acquisition, feedback, and UE-side computation overhead. \\
\indent To overcome this limitation, we propose a site-specific Type-II feedback scheme with two components: BS-side SSI-conditioned subspace inference and UE-side low-dimensional coefficient feedback. In particular, the BS first infers a UE-dependent dominant beam subspace from site-specific knowledge, and the UE then only refines the low-dimensional coefficients within that subspace. In this way, the proposed scheme reduces UE-side overhead while preserving the rich subspace representation capability of Type-II feedback.
\vspace{-0.4cm}
\subsection{Site-Specific Subspace Inference}
Compared with conventional Type-II feedback, the key idea of the proposed feedback scheme is to infer the dominant transmit subspace at the BS before explicit UE-side CSI estimation and feedback. To make this possible, the BS must rely on a lightweight channel-dependent observation that is available before the CSI-RS stage. The RSRP vector $\mathbf{r}_{\mathbf{B}}$ acquired during SSB probing serves this purpose. \\ \indent Specifically, the RSRP vector captures how the channel energy is distributed over a set of coarse probing directions. Although this RSRP vector is insufficient to recover the full instantaneous CSI because phase information is unavailable, it provides a beam-domain power fingerprint that is strongly related to the dominant angular support of the channel \cite{Heng2022SiteSpecificProbing,Ning2023RSRPCodebook,sim}. In site-specific environments with relatively stable propagation geometry and a limited number of dominant scattering clusters, similar dominant propagation directions tend to induce similar RSRP patterns. Accordingly, the RSRP vector provides a suitable low-overhead input for BS-side prediction of a candidate dominant beam subspace. \\
\indent To this end, we suppose that the BS infers a UE-dependent low-dimensional channel subspace through a mapping
\begin{equation}
\mathbf{C}_p = \Psi(\mathbf{r}_{\mathbf{B}}) = [\mathbf{c}_{p,1},\cdots,\mathbf{c}_{p,Q}] \in \mathbb{C}^{N_t \times Q},
\end{equation}
where $Q \ll N_t$ and the columns of $\mathbf{C}_p$ are expected to align with the dominant propagation directions of the current UE. Without loss of generality, we assume that $\mathbf{C}_p$ is an orthonormal basis of the inferred subspace, since any linearly independent basis can be transformed into an equivalent orthonormal basis via QR decomposition without changing its span. \\
\indent In this way, the dominant transmit directions are no longer searched exhaustively at the UE as in conventional Type-II, but are instead inferred at the BS by conditioning the offline learned SSI on the RSRP fingerprint. The remaining online feedback task is therefore reduced to refining the instantaneous low-dimensional coefficients over the inferred subspace.
\vspace{-0.3cm}
\subsection{Low-Dimensional Coefficient Feedback}
After the BS determines $\mathbf{C}_p$, it transmits beamformed CSI-RS only over this inferred subspace. The received CSI-RS at the UE is then
\begin{equation} \label{proposed CSI-RS}
	\mathbf{y}_{\text{CSI}}^T(\mathbf{C}_p) = \sqrt{P_{\text{CSI}}} \mathbf{h}^H \mathbf{C}_p \mathbf{S}_{\text{CSI}} + \mathbf{n}_{\text{CSI}}^T,
\end{equation}
where $\mathbf{S}_{\text{CSI}} \in \mathbb{C}^{Q \times L_c}$. Compared with the conventional limited-feedback schemes, the CSI-RS dimension is reduced from $N_t$ to $Q$, which significantly reduces the CSI acquisition overhead from $\mathcal{O}(N_t^2)$ to $\mathcal{O}(Q^2)$. In this case, the UE observes the effective channel over the inferred subspace
\begin{equation}
	\mathbf{h}_p = \mathbf{C}_p^H \mathbf{h}.
\end{equation}
We still assume that the UE can perfectly recover $\mathbf{h}_p$ from the received CSI-RS, i.e., $\mathbf{h}_p = \mathcal{R}_p\big(\mathbf{y}_{\text{CSI}}(\mathbf{C}_p)\big)$. \\
\indent Therefore, unlike conventional Type-II feedback, the UE no longer needs to either determine the dominant beam directions or estimate the coefficients by solving an optimization problem from the full-dimensional channel estimate. Instead, it can directly report the estimated $\mathbf{h}_p$, which already provides a low-dimensional representation of the DL CSI, to the BS. In this sense, the UE-side CSI compression mapping $\mathcal{Q}_p$ reduces to an identity mapping over the inferred low-dimensional channel domain:
\begin{equation} \label{zp}
	\mathbf{z}_p=\mathcal{Q}_p(\mathbf{h}_p)=\boldsymbol{\alpha}_p=\mathbf{h}_p.
\end{equation}
Accordingly, after receiving $\mathbf{z}_p$, the BS reconstructs the channel representation according to
\begin{equation} \label{hp}
	\hat{\mathbf{h}}_p=\mathcal{F}_p(\mathbf{z}_p;\mathbf{C}_p)=\mathbf{C}_p\mathbf{z}_p=\mathbf{C}_p\mathbf{h}_p=\mathbf{C}_p\mathbf{C}_p^H\mathbf{h}.
\end{equation}
In this process, the CSI representation subspace coincides with the CSI-RS precoder, i.e., $\mathbf{U}_p = \mathbf{C}_p$. The corresponding orthogonal projector is $\mathbf{P}_p = \mathbf{C}_p \mathbf{C}_p^H$. The beamforming vector is then obtained by applying MRT to the reconstructed channel as
\begin{equation} \label{wp}
\hat{\mathbf{w}}_p = \frac{\mathcal{F}_p\big(\mathcal{Q}_p(\mathcal{R}_p(\mathbf{y}_{\text{CSI}} (\mathbf{C}_p)))\big)}{\|\mathcal{F}_p\big(\mathcal{Q}_p(\mathcal{R}_p(\mathbf{y}_{\text{CSI}} (\mathbf{C}_p)))\big)\|}=\frac{\mathbf{P}_p \mathbf{h}}{\|\mathbf{P}_p \mathbf{h}\|}.
\end{equation}
In this way, the proposed scheme remains compatible with the existing feedback pipeline while shifting both the high-dimensional dominant-subspace discovery task and the following coefficient retrieval burden from the UE to the BS through site-specific subspace inference.
\vspace{-0.3cm}
\subsection{CSI-Capture Efficiency Analysis}
The achievable rate of the proposed design is given by
\begin{equation}
	R_p = \log_2\left(1+\rho\left| \mathbf{h}^H\hat{\mathbf{w}}_p \right|^2\right) = \log_2\left(1+\rho\|\mathbf{P}_p\mathbf{h}\|^2\right),
\end{equation}
Suppose that there exists an oracle $Q$-dimensional subspace $\mathcal{U}_Q^\star$ that captures the most channel energy among all $Q$-dimensional subspaces, i.e., 
\begin{equation}
	\mathcal{U}_Q^\star = \arg\max_{\mathcal{U}:\dim(\mathcal{U})=Q} \|\mathbf{P}_{\mathcal{U}}\mathbf{h}\|^2.
\end{equation}
The corresponding orthogonal projector is denoted by $\mathbf{P}_Q^\star$. Here, $\mathcal{U}_Q^\star$ is used only as an analytical benchmark, since such a subspace is practically inaccessible at the BS. For example, when $Q=1$, $\mathcal{U}_Q^\star=\mathrm{span}(\mathbf{h})$ is an oracle subspace, but the BS cannot infer $\mathbf{h}$ from the RSRP vector. Alternatively, this oracle subspace may be interpreted as the subspace that best aligns with the dominant propagation directions. In particular, the oracle subspace can be spanned by the array response vectors of the $Q$ strongest paths, i.e., $\mathcal{U}_Q^\star = \mathrm{span}([\mathbf{a}(u_l)]_{l\in\mathcal L_Q^\star})$, where $\mathcal L_Q^\star$ denotes the set of indices of the $Q$ strongest paths. The captured channel energy is therefore given by
\begin{equation}
	\|\mathbf{P}_Q^\star \mathbf{h}\|^2	= \boldsymbol{\alpha}^H\mathbf{A}^H\mathbf{P}_Q^\star \mathbf{A}\boldsymbol{\alpha}\approx	\sum_{l \in \mathcal L_Q^\star}|\alpha_l|^2 g_{l,Q^\star},
\end{equation}
where
\begin{equation}
	g_{l,Q^\star} \triangleq \|\mathbf{P}_Q^\star \mathbf{a}(u_l)\|^2 \approx 1, \forall l \in \mathcal L_Q^\star
\end{equation}
denotes the corresponding oracle path-capture factor. In this case, the top-$Q$ paths are fully captured by the oracle subspace without the grid mismatch induced by the DFT basis in Type-I/II. Note that this oracle subspace is determined by the specific application and is typically unavailable to the BS. It is used only for theoretical analysis.
\begin{theorem}[\emph{Subspace-mismatch bound for site-specific Type-II feedback}] \label{theo2} \normalfont
Let $\mathbf{P}_p$ be the projector onto the subspace inferred by the proposed SSI-enhanced feedback scheme, and let $\mathbf{P}_Q^\star$ denote the projector onto an oracle $Q$-dimensional benchmark subspace. Define $\delta_p\triangleq\|\mathbf{P}_p-\mathbf{P}_Q^\star\|_2$ and $[x]_+\triangleq\max\{x,0\}$. Then, the proposed scheme satisfies
\begin{equation}
	\|\mathbf{P}_p\mathbf{h}\|^2 \ge \left[\|\mathbf{P}_Q^\star\mathbf{h}\|^2-\delta_p\|\mathbf{h}\|^2\right]_+,
\end{equation}
and its achievable rate is lower bounded by
\begin{equation}
	R_p \ge \log_2\left(1+\rho\left[\|\mathbf{P}_Q^\star\mathbf{h}\|^2-\delta_p\|\mathbf{h}\|^2\right]_+\right).
\end{equation}
Moreover, the CSI-capture efficiency satisfies
\begin{equation}
	\eta_p = \frac{\|\mathbf{P}_p \mathbf{h}\|^2}{\|\mathbf{h}\|^2} \ge \left[\frac{\|\mathbf{P}_Q^\star\mathbf{h}\|^2}{\|\mathbf{h}\|^2} - \delta_p\right]_+.
\end{equation}
\end{theorem}
\vspace{-0.1cm}
\begin{IEEEproof}
Since both $\mathbf{P}_p$ and $\mathbf{P}_Q^\star$ are Hermitian projectors,
\begin{equation}
	\|\mathbf{P}_p\mathbf{h}\|^2 = \|\mathbf{P}_Q^\star\mathbf{h}\|^2 + \mathbf{h}^H(\mathbf{P}_p-\mathbf{P}_Q^\star)\mathbf{h}.
\end{equation}
The Rayleigh quotient bound gives $\mathbf{h}^H(\mathbf{P}_p-\mathbf{P}_Q^\star)\mathbf{h}\ge -\|\mathbf{P}_p-\mathbf{P}_Q^\star\|_2\|\mathbf h\|^2$. Combining this inequality with the non-negativity of $\|\mathbf{P}_p\mathbf h\|^2$ yields the energy bound. Substituting it into $R_p=\log_2(1+\rho\|\mathbf{P}_p\mathbf h\|^2)$ and normalizing by $\|\mathbf h\|^2$ complete the proof.
\end{IEEEproof}
\indent \textbf{Theorem~\ref{theo2}} separates the performance of the proposed feedback scheme into two interpretable terms. The first term, $\|\mathbf P_Q^\star\mathbf h\|^2$, represents the best channel energy that can be preserved by an ideal $Q$-dimensional benchmark subspace. It is therefore a dimension-limited benchmark: it is upper bounded by $\|\mathbf h\|^2$ and becomes equal to the full-channel energy only when the chosen $Q$-dimensional subspace contains the channel direction. The second term, $\delta_p\|\mathbf h\|^2$, quantifies the mismatch between the SSI-inferred subspace and the benchmark oracle subspace. Hence, the theorem shows that the proposed design is limited by both the intrinsic dimension constraint of the reduced CSI-RS subspace and the accuracy with which SSI-conditioned inference identifies that subspace. The overall illustration of the proposed limited-feedback scheme is shown in Fig. \ref{fig:arch}.
\begin{figure}[t]
	\centering
	\includegraphics[scale=0.48]{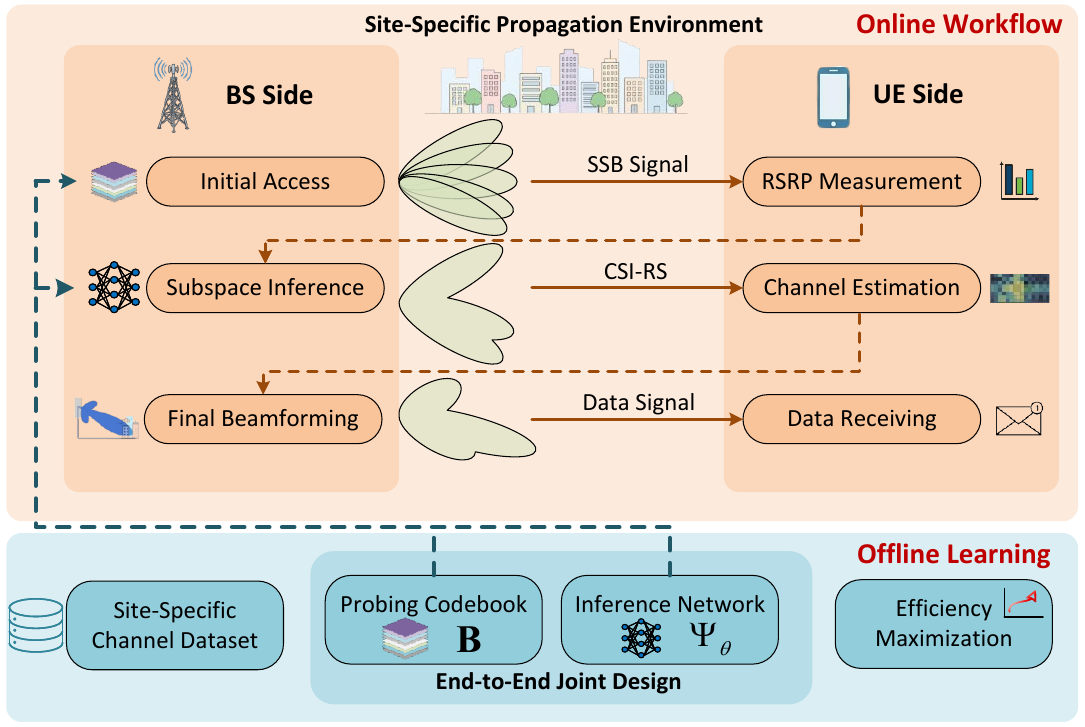}
	\caption{Illustration of the proposed site-specific Type-II feedback}
	\label{fig:arch}
\end{figure}
\vspace{-0.2cm}
\section{Joint Design of SSB Probing and Subspace Inference} \label{sec:realization}
The preceding section describes the proposed site-specific Type-II feedback scheme at the system and signal-processing level, where the site-specific subspace inference is abstracted as a mapping from the SSB probing measurement $\mathbf{r}_{\mathbf{B}}$ to a low-dimensional CSI subspace. The preceding analysis further demonstrates that the effectiveness of the proposed scheme critically depends on the quality of the inferred subspace $\mathbf{C}_p$, which is determined by the inference mapping $\Psi(\cdot)$.
\vspace{-0.45cm}
\subsection{Problem Formulation}
The subspace-mismatch bound above motivates a task-oriented design criterion: the inferred subspace should preserve as large a fraction of the channel energy as possible for the channel distribution of the target site under a $Q$-dimensional subspace bottleneck. However, directly minimizing the oracle mismatch $\delta_p$ is not practical because the oracle projector is unavailable and can be non-unique. We therefore formulate the design in terms of CSI-capture efficiency, which directly measures the relative energy preserved by the inferred subspace. \\
\indent The probing codebook $\mathbf B$ is coupled with this objective because the inferred subspace is computed only from the RSRP fingerprint $\mathbf r_{\mathbf B}$ generated during SSB probing. A generic oversampled DFT probing codebook provides uniform directional sounding, but it does not necessarily produce measurements that are most informative for dominant-subspace inference as analyzed and illustrated in \cite{sim}. Conversely, the information-maximizing probing design in \cite{sim} is not tailored to the subspace-capture objective considered here. Therefore, the probing stage and the BS-side inference mapping should be optimized as a single task-oriented module: the probing codebook should shape the RSRP fingerprint so that the inference mapping can recover a low-dimensional subspace with high CSI-capture efficiency. \\
\indent Accordingly, for a given channel realization, the joint design of the probing codebook $\mathbf{B}$ and the inference mapping $\Psi(\cdot)$ can be formulated as the following optimization problem
\begin{equation} \label{eq:single_eff_obj}
	\max_{\mathbf{B},\,\Psi}\;\eta_p(\mathbf h;\mathbf B,\Psi) = \left\|\Psi(\mathbf r_{\mathbf B})\Psi^H(\mathbf r_{\mathbf B})\mathbf h\right\|^2 / \|\mathbf h\|^2,
\end{equation}
where the inferred basis is assumed to have been orthogonalized before constructing the projector. However, seeking a probing codebook $\mathbf{B}$ and an inference mapping $\Psi(\cdot)$ such that the inferred subspace preserves the channel energy for every possible channel realization is overly stringent for the present low-dimensional probing design. In particular, when $Q<N_t$, satisfying the above condition for all $\mathbf{h}$ would require the probing measurement $\mathbf{r}_{\mathbf{B}}$ to identify, from low-overhead observations, a $Q$-dimensional subspace aligned with arbitrary directions in $\mathbb{C}^{N_t \times 1}$, which is in general impossible. Therefore, rather than pursuing universal optimality over the entire channel space, it is more practical and sufficient for the present problem to optimize the probing codebook and the inference mapping over the site-specific channel set induced by the target propagation environment:
\begin{equation} \label{eq:site_eff_obj}
\max_{\mathbf{B},\,\Psi}\; \mathbb{E}_{\mathbf{h} \sim p_s}\!\left[ \eta_p(\mathbf h;\mathbf B,\Psi) \right],
\end{equation}
where the expectation is taken over the site-specific channel distribution $p_s$.
\vspace{-0.45cm}
\subsection{Learning-based End-to-End Solver}
The above joint design problem remains difficult to solve by conventional optimization techniques. First, the site-specific channel distribution $p_s$ is generally unavailable in closed form and can only be accessed through channel samples collected from the target environment. Second, the probing codebook $\mathbf{B}$ and the inference mapping $\Psi(\cdot)$ are strongly coupled through the nonlinear RSRP measurement process and the resulting subspace projector. Third, practical constraints on the probing codebook, such as beam normalization or phase-only implementation, further render the problem highly nonconvex. Therefore, instead of seeking a closed-form solution, we adopt a task-driven end-to-end learning framework that directly optimizes the probing codebook and the inference mapping from site-specific channel data. \\
\indent Specifically, the probing codebook $\mathbf{B}$ serves as an encoder $f_e$, which maps the input channel to a low-dimensional dB-domain RSRP fingerprint,
\begin{equation} \label{fe}
	\mathbf{r}_{\mathbf{B}}^0 = f_e(\mathbf{h};\mathbf{B})= 10\log_{10}\!\left(P_{\mathrm{SSB}}\left|\mathbf{B}^H\mathbf{h}\right|^2\right),
\end{equation}
where the magnitude square and logarithm are applied elementwise. The intermediate layer then models the measurement uncertainty introduced by shadowing and thermal noise according to (\ref{RSRP}), yielding the observed RSRP vector $\mathbf{r}_{\mathbf{B}}$. The subspace inference mapping is parameterized by a neural network $\Psi_{\theta}(\cdot)$ with trainable parameters $\boldsymbol{\theta}$ and serves as the decoder $f_d$, which takes the noisy RSRP measurement as input and produces the inferred subspace representation,
\begin{equation} \label{fd}
	\mathbf{C}_p = f_d(\mathbf{r}_{\mathbf{B}};\boldsymbol{\theta}) = \Psi_{\theta}(\mathbf{r}_{\mathbf{B}}).
\end{equation}
\indent Let $\mathcal{H}_s=\{\mathbf{h}^{(n)}\}_{n=1}^{N_h}$ denote a set of channel samples collected from the target site. If $\{\mathbf{h}^{(n)}\}_{n=1}^{N_h}$ are i.i.d. samples drawn from the site-specific channel distribution $p_s$, then by the law of large numbers the empirical measure induced by $\mathcal{H}_s$ converges to $p_s$ as $N_h\to\infty$. In this case, the statistical objective in (\ref{eq:site_eff_obj}) is approximated by the empirical average
\begin{equation} \label{sample_average}
\min_{\mathbf{B},\boldsymbol{\theta}}\;\mathcal{L}_{\mathcal{H}_s} = -\frac{1}{N_h}\sum_{n=1}^{N_h} \eta_p\!\left(\mathbf h^{(n)};\mathbf B,\Psi_\theta\right),
\end{equation}
where $\eta_p\!\left(\mathbf h^{(n)};\mathbf B,\Psi_\theta\right)$ is evaluated using the probing measurement $\mathbf{r}_{\mathbf{B}}^{(n)}$ generated from $\mathbf{h}^{(n)}$ under the current probing codebook $\mathbf{B}$. \\
\indent A useful perspective is to view the proposed scheme as the first stage of a hierarchical CSI autoencoding process. Instead of directly reconstructing the instantaneous channel from a learned latent representation, the probing encoder and the subspace decoder first encode and recover a structural representation of the channel, namely, its dominant low-dimensional subspace. In this sense, the site-specific probing stage acts as a coarse channel autoencoder that extracts the slowly-varying spatial structure from the full-dimensional channel and represents it by a compact RSRP fingerprint. The subsequent CSI feedback stage then performs a second-stage low-dimensional refinement by feeding back the instantaneous coefficients within the inferred subspace. Therefore, the proposed framework decomposes CSI acquisition into structural subspace encoding and residual coefficient encoding, which is precisely what enables low-overhead yet effective feedback.
\vspace{-0.5cm}
\subsection{Practical Realization}
Since the probing encoder $\mathbf{B}$ is inherently trainable, it can be directly optimized in the learning process, subject only to the unit-power constraint $\|\mathbf{b}_k\|^2=1$ for each probing beam. The subspace decoder $\Psi_{\theta}$ can be implemented by a multi-layer perceptron (MLP), which has a depth-$D$ fully connected architecture with hidden width $W$, layer normalization (LN) \cite{ba2016layer}, and Gaussian error linear unit (GELU) activation \cite{Hendrycks2016GELU}. Let $\mathbf{f}^{(0)}=\bar{\mathbf{r}}_{\mathbf{B}}^{(n)}$. Then, for $d=1,\ldots,D$, the hidden features are updated as
\begin{equation}
	\mathbf{f}^{(d)} = \mathrm{GELU}\!\left(\mathrm{LN}\!\left(\mathbf{W}_d\mathbf{f}^{(d-1)} + \mathbf{b}_d\right)\right) \in \mathbb{R}^{W \times 1},
\end{equation}
where $\mathrm{GELU}(x) = x \Phi(x)$ and $\Phi(x)$ is the cumulative distribution function of the standard Gaussian distribution. The output layer yields
\begin{equation}
	\mathbf{f}^{(\mathrm{out})} = \mathbf{W}_{\rm out}\mathbf{f}^{(D)} + \mathbf{b}_{\rm out} \in \mathbb{R}^{2N_tQ \times 1}.
\end{equation}
The output vector is then reshaped into the real and imaginary parts of a raw complex $Q$-dimensional subspace representation $\widetilde{\mathbf{C}}_p$. To remove the scale ambiguity of the predicted basis, each column of $\widetilde{\mathbf{C}}_p$ is normalized to unit norm and orthogonalized before constructing the corresponding subspace projector. \\
\indent With the differentiable RSRP measurement model adopted above, the sample-average objective in (\ref{sample_average}) can be optimized with respect to both $\mathbf{B}$ and $\boldsymbol{\theta}$ through backpropagation. However, accurately capturing the site distribution typically requires a large $N_h$, which makes full-batch optimization computationally burdensome. Mini-batch stochastic gradient descent (SGD) \cite{GhadimiLan2013} offers a scalable alternative by using a randomly sampled subset of UEs to form a stochastic approximation of the objective and its gradient. For a mini-batch $\mathcal{B}\subset\mathcal{H}_s$, the empirical training loss is defined as
\begin{equation} \label{loss}
\mathcal{L}_{\mathcal{B}}(\mathbf{B},\boldsymbol{\theta}) = -\frac{1}{|\mathcal{B}|} \sum_{\mathbf{h}^{(n)}\in\mathcal{B}} \eta_p\!\left(\mathbf h^{(n)};\mathbf B,\Psi_\theta\right).
\end{equation}
\indent The overall implementation then follows a two-stage procedure, namely offline training and online deployment. During offline training, the probing codebook $\mathbf{B}$ and the inference network $\Psi_{\theta}(\cdot)$ are optimized jointly over site-specific channel samples, which is summarized in \textbf{Algorithm~\ref{alg:joint_training}}. During online deployment, the learned probing codebook $\mathbf{B}$ and subspace inference mapping $\Psi_{\theta}(\cdot)$ are embedded into the framework introduced in Section \ref{sec:proposed}, whose detailed procedure is summarized in \textbf{Algorithm~\ref{alg:online_deployment}}.
\begin{algorithm}[t]
	\setlength{\textfloatsep}{0.cm}
	\setlength{\floatsep}{0.cm}
	\small
	\renewcommand{\algorithmicrequire}{\textbf{Input}}
	\renewcommand{\algorithmicensure}{\textbf{Output}}
	\caption{Offline Joint Training of the Probing Codebook and the Subspace Inference Mapping}
	\label{alg:joint_training}
	\begin{algorithmic}[1]
		\REQUIRE Site-specific CSI dataset $\mathcal{H}_s$, batch size $B$, step $\beta$, iterations $I$
		\ENSURE Optimized probing codebook and subspace inference mapping $\mathbf{B}^{\star}$ and $\Psi_{\theta^{\star}}$
		\STATE Initialize probing codebook $\mathbf{B}^{(0)}$ and network parameters $\boldsymbol{\theta}^{(0)}$.
		\FOR{$i=1,2,\cdots,I$}
		\STATE Sample a mini-batch of data $\mathcal{B}_i\subset\mathcal{H}_s$ with $|\mathcal{B}_i|=B$.
		\STATE Generate probing RSRP measurements by (\ref{fe}).
		\STATE Contaminate the RSRP measurements by (\ref{RSRP}).
		\STATE Infer the subspace by (\ref{fd}).
		\STATE Compute the loss $\mathcal{L}_{\mathcal{B}_i}(\mathbf{B}^{(i-1)},\boldsymbol{\theta}^{(i-1)})$ by (\ref{loss}).
		\STATE Update $\mathbf{B}$ and $\boldsymbol{\theta}$ by $\mathbf{B}^{(i)}\leftarrow\mathbf{B}^{(i-1)}-\beta\nabla_{\mathbf{B}}\mathcal{L}_{\mathcal B_i}$ and $\boldsymbol{\theta}^{(i)}\leftarrow\boldsymbol{\theta}^{(i-1)}-\beta\nabla_{\boldsymbol{\theta}}\mathcal{L}_{\mathcal B_i}$.
		\ENDFOR
		\STATE The trained probing codebook and inference module are obtained as $\mathbf{B}^{\star} = \mathbf{B}^{(I)}$ and $\Psi_{\theta^{\star}} = \Psi_{\theta}^{(I)}$.
	\end{algorithmic}
\end{algorithm}

\begin{algorithm}[t]
	\setlength{\textfloatsep}{0.cm}
	\setlength{\floatsep}{0.cm}
	\small
	\renewcommand{\algorithmicrequire}{\textbf{Input}}
	\renewcommand{\algorithmicensure}{\textbf{Output}}
	\caption{Online Deployment of the Proposed SSI-Enhanced Type-II Feedback}
	\label{alg:online_deployment}
	\begin{algorithmic}[1]
		\REQUIRE Trained probing codebook and inference mapping $\mathbf{B}^{\star}$ and $\Psi_{\theta^{\star}}$
		\ENSURE Beamforming vector $\hat{\mathbf{w}}_p$
		\STATE BS sweeps the learned probing codebook $\mathbf{B}^{\star}$ with SSB.
		\STATE UE feeds back the measured RSRP vector $\mathbf{r}_{\mathbf{B}}$.
		\STATE BS infers the dominant subspace $\mathbf{C}_p$ by (\ref{fd}).
		\STATE BS transmits low-dimensional beamformed CSI-RS over the inferred subspace by (\ref{proposed CSI-RS}).
		\STATE UE feeds back the estimated effective channel by (\ref{zp}).
		\STATE BS reconstructs the DL CSI by (\ref{hp}) and applies the beamforming vector by (\ref{wp}).
	\end{algorithmic}
\end{algorithm}
\begin{table}[t]
	\centering
	\caption{Summary of Conventional and Proposed Limited-Feedback Schemes}
	\label{tab:scheme_summary}
	\renewcommand{\arraystretch}{1.1}
	\footnotesize
	\resizebox{\columnwidth}{!}{%
		\begin{tabular}{c c c c c}
			\toprule
			Scheme & Subspace $\mathcal U$ & Overhead $T_o$ & UE Complexity & Efficiency $\eta$ \\
			\midrule
			Type-I & $\mathrm{span}(\mathbf d_{\mathbf z_{\rm I}})$ & $N_t+1$ & $\mathcal{O}(N_t^2O_D)$ & $\eta_{\rm I}\!\approx\! \frac{|\alpha_{l_{\rm I}}|^2 g_{l_{\rm I}}} {\sum_{l=1}^{L}|\alpha_l|^2}$ \\
			Type-II & $\mathrm{span}(\mathbf D_{\mathcal S_{\rm II}})$ & $N_t+2Q$ & $\mathcal{O}(N_t^2QO_D+Q^3)$ & $\eta_{\rm II}\!\approx\! \frac{\sum_{l\in\mathcal L_{Q,\rm II}}|\alpha_l|^2 g_{l,\rm II}} {\sum_{l=1}^{L}|\alpha_l|^2}$ \\
			PSC & $\mathrm{span}(\mathbf E_{\mathcal S_{\rm PSC}})$ & $N_t+2N_p$ & $\mathcal{O}(N_t^2)$ & $\eta_{\rm PSC}\!=\! \frac{\|\mathbf{P}_{\rm PSC}\mathbf{h}\|^2}{\|\mathbf{h}\|^2} = \frac{\sum_{i=1}^{N_p} |\tilde h|_{(i)}^2}{\|\mathbf{h}\|^2}$ \\
			Proposed & $\mathrm{span}(\mathbf C_p)$ & $K+2Q$ & $\mathcal{O}(Q^2)$ & $\eta_{\rm p}\!=\! \frac{\|\mathbf P_{\rm p}\mathbf h\|^2}{\|\mathbf h\|^2}$ \\
			\bottomrule
	\end{tabular}}
\end{table}
\subsection{Convergence and Complexity Analysis}
Let $\boldsymbol{\xi}\triangleq(\mathbf{B},\boldsymbol{\theta})$ collect all trainable variables. Suppose each mini-batch $\mathcal{B}_i$ is uniformly sampled from $\mathcal{H}_s$, and define the stochastic gradient
\begin{equation}
	\mathbf{g}_i \triangleq \nabla_{\boldsymbol{\xi}} \mathcal{L}_{\mathcal{B}_i}(\boldsymbol{\xi}_i).
\end{equation}
Assume that $\mathcal{L}_{\mathcal{H}_s}(\boldsymbol{\xi})$ is lower bounded and $L_f$-smooth, and that the stochastic gradient is an unbiased estimator of the full gradient with bounded variance that decreases with the mini-batch size, i.e.,
\begin{equation}
	\mathbb{E}[\mathbf{g}_i]=\nabla_{\boldsymbol{\xi}} \mathcal{L}_{\mathcal{H}_s}(\boldsymbol{\xi}_i), \ 
	\mathbb{E}\!\left[\|\mathbf{g}_i-\nabla_{\boldsymbol{\xi}} \mathcal{L}_{\mathcal{H}_s}(\boldsymbol{\xi}_i)\|^2\right] \le \frac{\sigma^2}{|\mathcal{B}_i|}.
\end{equation}
Then, the mini-batch SGD update $\boldsymbol{\xi}_{i+1} = \boldsymbol{\xi}_i-\gamma_i\mathbf{g}_i$ with a standard diminishing or sufficiently small constant stepsize converges in expectation to a first-order stationary point. In particular, for the usual $\mathcal O(1/\sqrt{I})$ stepsize scaling, the average gradient norm obeys the standard nonconvex SGD rate
\begin{equation}
	\frac{1}{I}\sum_{i=0}^{I-1}\mathbb{E}\!\left[\|\nabla_{\boldsymbol{\xi}} \mathcal{L}_{\mathcal{H}_s}(\boldsymbol{\xi}_i)\|^2\right]	= \mathcal{O}\!\left(\frac{1}{\sqrt{I}}+\frac{\sigma^2}{|\mathcal B|\sqrt{I}}\right).
\end{equation}
This result follows from the standard descent analysis for smooth nonconvex stochastic optimization and shows that the proposed mini-batch joint solver approaches a stationary solution in expectation, with the gradient-variance term reduced by the mini-batch size. \\
\indent We focus on the online deployment complexity, since the learning-based design is trained offline and its cost is amortized over the deployment site. During online deployment, the main complexity of the proposed scheme comes from BS-side subspace inference and UE-side effective channel estimation. For a depth-$D$ MLP with hidden width $W$, the BS-side inference complexity is $\mathcal{O}\!\left(KW + (D-1)W^2 + 2N_tQW\right)$, while the UE-side LS estimation complexity is on the order of $\mathcal{O}(Q^2)$. By contrast, conventional feedback schemes place the dominant CSI-processing burden at the UE. Specifically, the UE-side complexity scales as $\mathcal{O}(N_t^2O_D)$ for Type-I, $\mathcal{O}(N_t^2QO_D+Q^3)$ for Type-II with greedy beam selection, and $\mathcal{O}(N_t^2+N_t\log N_t)$ for PSC, whereas the corresponding BS-side reconstruction complexity is only $\mathcal{O}(N_t)$, $\mathcal{O}(N_tQ)$, and $\mathcal{O}(N_tN_p)$, respectively. \\
\indent Table~\ref{tab:scheme_summary} summarizes the conventional Type-I, Type-II, PSC, and proposed schemes under the unified subspace-projection viewpoint. Since SSB sweeping is shared by all schemes, it is excluded from the overhead comparison, and only the additional SSB-stage reporting required by the proposed scheme is counted. As shown in Table~\ref{tab:scheme_summary}, the proposed site-specific Type-II retains Type-II-like subspace representation capability with only low-dimensional online refinement, thereby achieving a more favorable overhead-complexity-performance tradeoff.

\vspace{-0.5cm}
\section{Numerical Results} \label{sec:simulation}
In this section, numerical results are provided to verify the effectiveness of the proposed end-to-end joint probing-and-inference design and the proposed site-specific Type-II feedback framework based on the DeepMIMO dataset \cite{Alkhateeb2019DeepMIMO}. Key simulation parameters are summarized in Table~\ref{tab:sim_setup_new}, unless otherwise stated. By default, we use the ``asu\_campus\_3p5'' scenario, a 3.5 GHz outdoor deployment with rich multipath, moderate angular spread, and abundant UE samples. We also use the ``O1\_28'', ``O1B\_28'', and ``boston5g\_3p5'' scenarios for additional validation, which correspond to a 28 GHz simple line-of-sight (LoS) outdoor scenario, a 28 GHz blocked non-LoS (NLoS) outdoor scenario, and a 3.5 GHz complex city scenario, respectively.
\begin{table}[t]
	\small
	\centering
	\caption{Simulation settings}
	\label{tab:sim_setup_new}
	\begin{tabular}{ccc}
		\hline
		\textbf{Parameter}       & \textbf{Description}                           & \textbf{Value}         \\ \hline
		$f_c$                        & Carrier frequency                            & 3.5 GHz                      \\
		$\text{BW}$                        & Bandwidth                            & 10 MHz                      \\
		$P_t$                        & Transmit power                            & 40 dBm                     \\
		$K$                        & Size of SSB codebook                            & 8                      \\
		$Q$                        & Dimension of the inferred subspace                            & 4                      \\
		$N_t$                        & Number of BS antennas                             & 64                      \\
		$L_s$                        & Number of SSB symbols                 & 20                     \\
		$T_c$                        & Coherence block length                 & 1000                     \\
		$d$                   & Antenna spacing                 & $\lambda/2$               \\
		$S_n$                        & Noise power spectrum density & -170 dBm/Hz                      \\
		$\sigma_\text{sh}^2$                        & Log-variance of shadowing  & 1 dB                      \\ 
		$D$                        & MLP depth  & 3                      \\ 
		$W$                        & Hidden width  & 256                      \\ 
		\hline
	\end{tabular}
\end{table}
\begin{table}[t]
	\centering
	\caption{CSI-capture efficiency comparison of different probing schemes across representative DeepMIMO scenarios.}
	\label{ablation}
	\renewcommand{\arraystretch}{1.1}
	\setlength{\tabcolsep}{6pt}
	\footnotesize
	\begin{tabular}{lccc}
		\toprule
		Scenario & Proposed & Random & DFT \\
		\midrule
		O1\_28            & \textbf{0.99} & 0.90 & 0.91 \\
		O1B\_28           & \textbf{0.99} & 0.93 & 0.90 \\
		asu\_campus\_3p5  & \textbf{0.89} & 0.62 & 0.73 \\
		boston5g\_3p5     & \textbf{0.97} & 0.81 & 0.86 \\
		\midrule
		Average           & \textbf{0.96} & 0.82 & 0.85 \\
		\bottomrule
	\end{tabular}
\end{table}
\vspace{-0.5cm}
\subsection{Evaluation of the End-to-End Joint Design}
\subsubsection{Convergence Demonstration}
We first demonstrate the convergence of the proposed mini-batch SGD solver for the joint design problem. Fig.~\ref{convergence} illustrates the convergence behavior of the proposed end-to-end solver under three representative $(K,Q)$ settings. The curves show stable convergence for all configurations, and most of the performance gain is achieved within roughly the first $100$ epochs, after which the curves enter a clear saturation region. Moreover, the training and validation curves remain close to each other for all three settings, indicating that the learned probing-and-inference model generalizes well and does not suffer from severe overfitting. In addition, the converged capture efficiency increases consistently as the probing dimension and the inferred subspace dimension become larger. In particular, the $(K,Q)=(16,8)$ configuration achieves the highest efficiency, followed by $(16,4)$ and $(8,4)$, which confirms that both richer probing measurements and a higher-dimensional inferred subspace are beneficial to the proposed framework. 
\begin{figure}[t]
	\centering
	\includegraphics[scale=0.45]{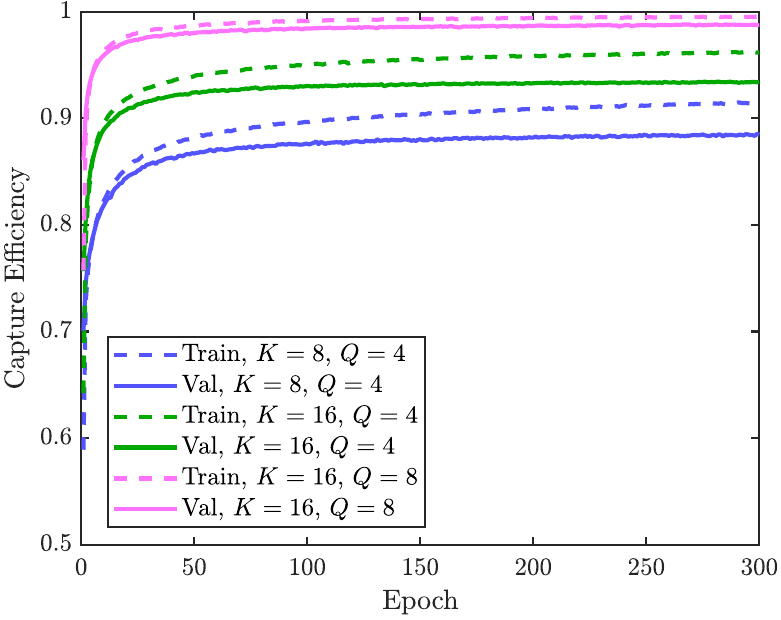}
	\caption{Convergence of the proposed end-to-end design}
	\label{convergence}
\end{figure}
\begin{figure}[t]
	\centering
	\includegraphics[scale=0.45]{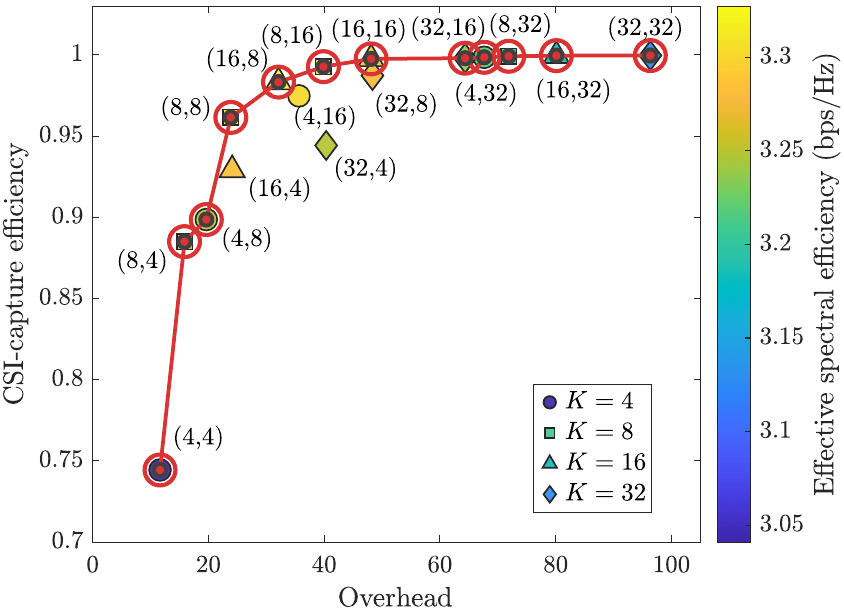}
	\caption{Pareto-optimal overhead-performance illustration}
	\label{tradeoff}
\end{figure}

\subsubsection{Ablation Study of the End-to-End Design}
To validate whether the gain of the proposed end-to-end solver indeed comes from the jointly learned probing encoder, we compare the CSI-capture efficiency of the proposed end-to-end design with two baseline probing modes under the same MLP decoder and training pipeline: i) a fixed random probing codebook, and ii) a fixed DFT probing codebook. Table~\ref{ablation} reports the comparison across four representative DeepMIMO scenarios. The table shows that the proposed scheme consistently achieves the highest capture efficiency in all scenarios. Concretely, in the relatively easier ``O1\_28'' and ``O1B\_28'' scenarios, all methods already achieve high capture efficiencies, while the proposed design still preserves a clear advantage. In more challenging environments, especially ``asu\_campus\_3p5'', the gain becomes substantially larger, indicating that the learned probing codebook is more effective at generating informative fingerprints for downstream subspace inference when the propagation structure is more complex. Moreover, the relative ordering between the random and DFT baselines changes across scenarios, whereas the proposed design remains consistently superior, which further demonstrates its robustness and site-adaptive nature.

\subsubsection{Overhead-Performance Tradeoff}
The previous results already indicate that the probing dimension $K$ and the inferred subspace dimension $Q$ are the two critical design parameters of the proposed framework, as they determine the informativeness of the RSRP fingerprint and the expressiveness of the inferred CSI subspace, respectively. Fig.~\ref{tradeoff} further characterizes the resulting overhead-performance tradeoff in the ``asu\_campus\_3p5'' scenario, where the horizontal axis denotes the total online overhead and the vertical axis denotes the CSI-capture efficiency. The Pareto-optimal operating points are highlighted by the red circles and the connecting curve. The capture efficiency improves rapidly in the low-overhead regime, but soon enters a saturation region as the overhead increases, revealing a pronounced diminishing-return effect. More importantly, the Pareto frontier shows that high performance is achieved only when $K$ and $Q$ are jointly balanced: a small $K$ limits the discriminability of the RSRP fingerprint, while a small $Q$ limits the representational capability of the inferred subspace. Hence, increasing only one of the two dimensions is generally inefficient once the other becomes the bottleneck. This explains why several moderate-overhead configurations already approach the maximum capture efficiency, whereas further enlarging both $K$ and $Q$ brings only marginal gain. In addition, the effective spectral efficiency, indicated by the color map, is not maximized at the largest-overhead points, since the incremental subspace-capture gain is eventually outweighed by the increased overhead penalty.
\vspace{-0.5cm}
\subsection{Evaluation of the Proposed Feedback Scheme}
\begin{table}[t]
	\centering
	\caption{CSI-capture efficiency comparison of different feedback schemes in two representative scenarios.}
	\label{tab:scheme_compare_two_scenarios}
	\renewcommand{\arraystretch}{1.1}
	\setlength{\tabcolsep}{8pt}
	\footnotesize
	\begin{tabular}{lcc}
		\toprule
		Scheme & Boston & O1\_28 \\
		\midrule
		Type-I    & 0.7017 & 0.7405 \\
		Type-II   & \textbf{0.9951} & 0.9992 \\
		PSC       & 0.2606 & 0.2510 \\
		Proposed  & 0.986 & \textbf{0.9999} \\
		\bottomrule
	\end{tabular}
\end{table}
\begin{figure}[t]
	\centering
	\includegraphics[scale=0.45]{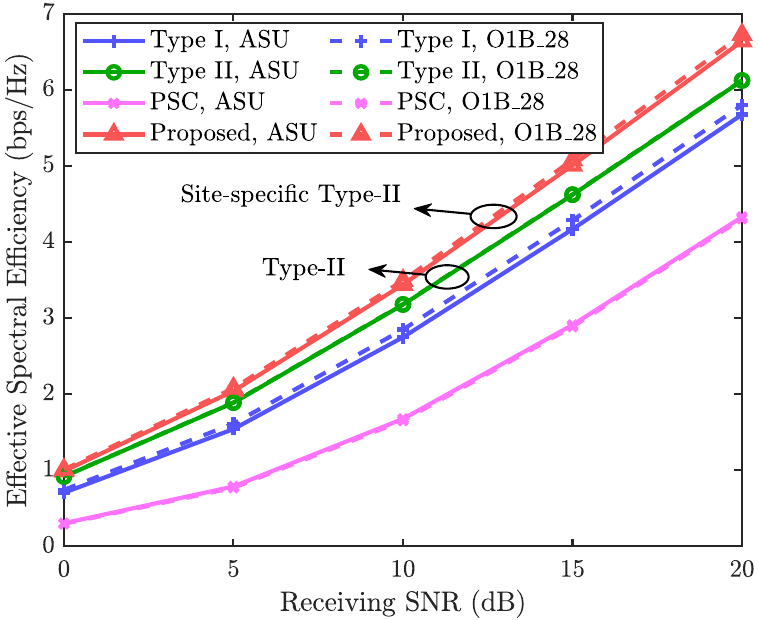}
	\caption{Effective spectral efficiency versus SNR}
	\label{comparison}
\end{figure}
\subsubsection{Scheme Comparison}
Table~\ref{tab:scheme_compare_two_scenarios} and Fig.~\ref{comparison} compare the proposed scheme configured with $(K, Q)=(16,8)$ against the Type-I, Type-II, and PSC baselines in the ``boston5g\_3p5'' and ``O1\_28'' scenarios under the same simulation settings. As expected, the proposed scheme attains a capture efficiency that is very close to, and in some cases even slightly higher than, that of Type-II. More importantly, despite this comparable raw efficiency, the proposed scheme consistently achieves the highest effective spectral efficiency over the entire SNR range. This result directly reflects the main advantage of the proposed SSI-enhanced framework: instead of pursuing marginal gains in instantaneous subspace optimality at the cost of full-dimensional CSI acquisition, it leverages site-specific inference to obtain Type-II-comparable subspace quality with much lower online overhead. Consequently, the proposed scheme converts Type-II-comparable capture efficiency into a strictly better system-level overhead-efficiency tradeoff. In other words, the key gain of the proposed framework comes not from universally outperforming Type-II in raw efficiency, but from matching Type-II in the relevant subspace-capture regime more efficiently.

\subsubsection{UE Performance Statistics}
Fig.~\ref{cdf} shows the cumulative distribution function (CDF) of the effective spectral efficiency of four feedback schemes at an SNR of 10 dB in the ``asu\_campus\_3p5'' and ``O1B\_28'' scenarios. The proposed scheme exhibits the rightmost CDF in both scenarios, indicating that it delivers higher effective spectral efficiency for the majority of users rather than only improving the average performance. This result is consistent with the previous average-rate comparison and further confirms that the gain of the proposed framework is population-wide. In particular, although the proposed scheme and Type-II often achieve comparable raw CSI-capture efficiency, the lower online overhead of the proposed framework shifts its effective-spectral-efficiency distribution consistently to the right of Type-II. Therefore, the benefit of the proposed method does not come from universally dominating Type-II in instantaneous subspace quality, but from realizing Type-II-comparable subspace quality with much lower overhead and converting this advantage into better UE-level effective spectral efficiency. By contrast, Type-I remains limited by its rank-one representation, which leads to a visibly broader and left-shifted distribution, while PSC performs the worst in both scenarios due to its insufficient ability to capture the dominant channel energy in the considered port-domain setting.
\begin{figure}[t]
	\centering
	\includegraphics[scale=0.45]{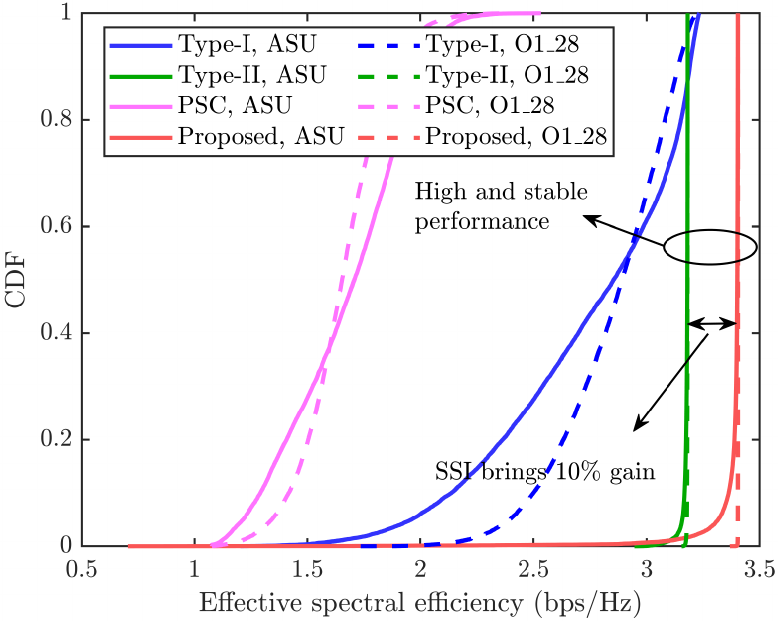}
	\caption{CDF of effective spectral efficiency}
	\label{cdf}	
\end{figure}

\subsubsection{Physical Interpretation of the Inferred Subspace}
Beyond rate and capture metrics, we further visualize the physical meaning of the learned subspace. Specifically, we select three representative test users in the ``asu\_campus\_3p5'' scenario and compare the angular responses of the proposed, Type-I, and Type-II subspace projectors. For each scheme $s\in\{p,\mathrm{I},\mathrm{II}\}$, let $\mathbf{P}_s$ denote the corresponding subspace projector. Its angular response is evaluated over the azimuth angle $\varphi$ as
\begin{equation}
	G_s(\varphi)=\mathbf{a}^H(\varphi)\mathbf{P}_s\mathbf{a}(\varphi).
\end{equation}
For each user, the three responses and path powers are all normalized by their maximum value for visualization.
\begin{figure*}[t]
	\centering
	\subfloat[User 1, max path power = -134.58 dB]{
		\label{phy1}
		\includegraphics[scale=0.31]{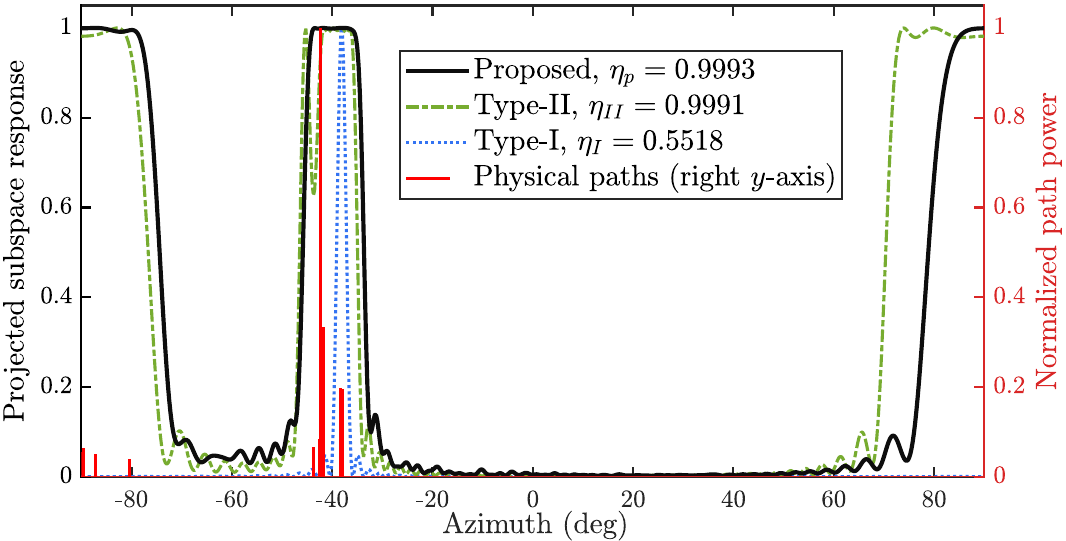}
	}
	\subfloat[User 2, max path power = -105.52 dB]{
		\label{phy2}
		\includegraphics[scale=0.31]{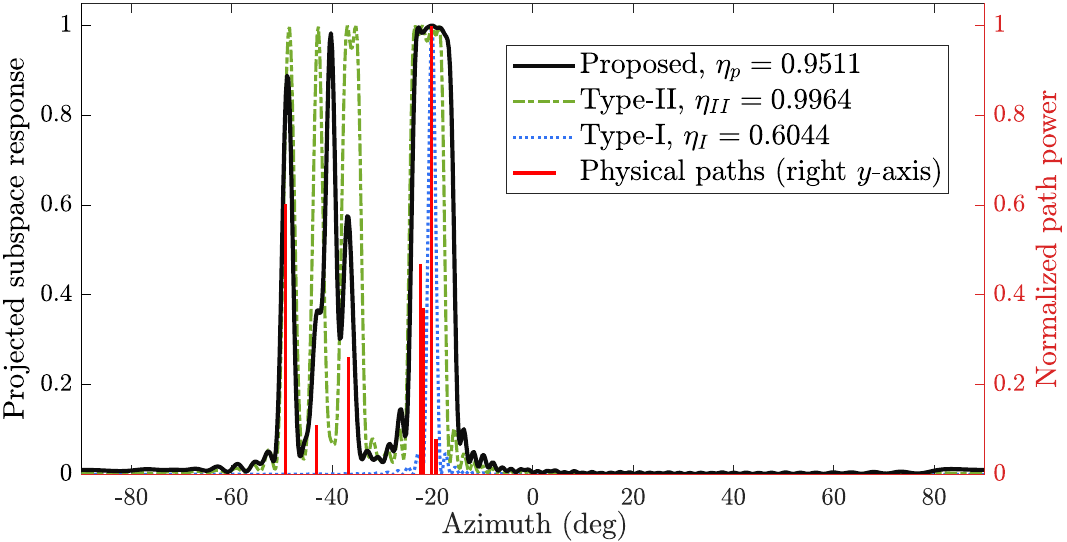}
	}
	\subfloat[User 3, max path power = -132.60 dB]{
		\label{phy3}
		\includegraphics[scale=0.31]{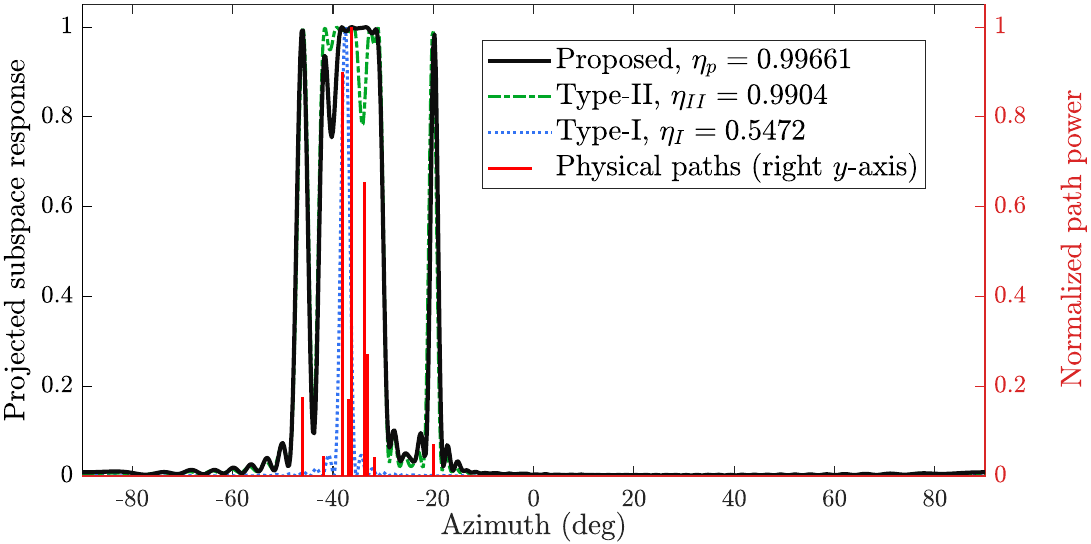}
	}
	\caption{Physical interpretation of the inferred subspace}
	\label{phy}
\end{figure*}

Fig.~\ref{phy} visualizes these angular responses. First, the response of the proposed scheme is consistently concentrated around the dominant physical path clusters, which confirms that the inferred subspace is not an arbitrary latent representation but has a clear geometric interpretation in the angular domain. Second, the proposed response closely resembles that of Type-II in the dominant angular regions, which explains why the proposed scheme achieves CSI-capture efficiency that is close to Type-II in the quantitative results. Third, the Type-I response is much narrower and typically covers only one local direction, which explains its substantially lower capture efficiency. Overall, these figures provide physical evidence that the proposed SSI-enhanced framework succeeds in inferring a low-dimensional subspace aligned with the dominant propagation geometry, thereby achieving Type-II-comparable subspace quality with significantly reduced overhead.

\section{Conclusion} \label{sec:conclusion}
This paper developed a unified subspace-projection framework for limited-feedback beamforming and proposed a site-specific Type-II feedback scheme. By using offline learned SSI together with low-overhead RSRP fingerprints, the BS infers a UE-dependent dominant subspace in advance, so that the UE only needs to estimate and feed back low-dimensional effective CSI coefficients within that subspace. Simulation results showed that the proposed scheme achieves a more favorable overhead-performance tradeoff than conventional feedback mechanisms, retaining Type-II-comparable subspace quality with substantially lower online CSI acquisition overhead and UE-side complexity. Future work may extend the framework to multi-user and multi-stream settings.

\balance
\bibliographystyle{IEEEtran}
\bibliography{reference/mybib}

\end{document}